\documentclass[lettersize,journal]{IEEEtran}
\usepackage{cite}
\usepackage{amsmath,amssymb,amsfonts,amsthm}
\usepackage{xcolor}
\usepackage{algorithm,algpseudocode}
\usepackage{array}
\usepackage{textcomp}
\usepackage{stfloats}
\usepackage{url}
\usepackage{soul}
\usepackage{graphicx}
\usepackage{wrapfig,lipsum}
\usepackage{subfigure}
\usepackage{balance}
\usepackage{enumitem} 
\usepackage{multirow}
\usepackage{mathrsfs}
\usepackage{diagbox}
\usepackage{xcolor}
\usepackage{nameref}
\usepackage[utf8]{inputenc}
\usepackage{booktabs}
\usepackage{caption}

\newtheorem{theorem}{Theorem}
\newtheorem{lemma}{Lemma}

\newtheorem{definition}{Definition}

\allowdisplaybreaks
\newcommand{\target}{\mathcal{T}}
\newcommand{\anom}{\mathcal{A}}
\newcommand{\threshold}{\Theta}
\newcommand{\simil}{\vec{s}}
\newcommand{\budget}{K}
\newcommand{\FA}{\widehat{A}}
\newcommand{\FR}{\widehat{R}}

\title{Coupled-Space Attacks against Random-Walk-based Anomaly Detection}

\author{Yuni Lai, Marcin Waniek, Liying Li, Jingwen Wu, Yulin Zhu, Tomasz P. Michalak, Talal Rahwan, Kai Zhou
  \thanks{Yuni Lai, Liying Li, Jingwen Wu, Yulin Zhu, and Kai Zhou are with the Department of Computing, The Hong Kong Polytechnic University. Marcin Waniek and Talal Rahwan are with the Department of Computer Science, New York University Abu Dhabi. Tomasz P. Michalak is with the Department of Computer Science, University of Warsaw and IDEAS NCBR. Kai Zhou is the corresponding author with email: kaizhou@polyu.edu.hk.}
}

\begin{document}
\maketitle

\begin{abstract}


Random Walks-based Anomaly Detection (RWAD) is commonly used to identify anomalous patterns in various applications. An intriguing characteristic of RWAD is that the input graph can either be pre-existing or constructed from raw features. Consequently, there are two potential attack surfaces against RWAD: graph-space attacks and feature-space attacks. In this paper, we explore this vulnerability by designing practical coupled-space attacks, investigating the interplay between graph-space and feature-space attacks. To this end, we conduct a thorough complexity analysis, proving that attacking RWAD is NP-hard. Then, we proceed to formulate the graph-space attack as a bi-level optimization problem and propose two strategies to solve it: alternative iteration (alterI-attack) or utilizing the closed-form solution of the random walk model (cf-attack). Finally, we utilize the results from the graph-space attacks as guidance to design more powerful feature-space attacks (i.e., graph-guided attacks). Comprehensive experiments demonstrate that our proposed attacks are effective in enabling the target nodes from RWAD with a limited attack budget. In addition, we conduct transfer attack experiments in a black-box setting, which show that our feature attack significantly decreases the anomaly scores of target nodes. Our study opens the door to studying the coupled-space attack against graph anomaly detection in which the graph space relies on the feature space.

\end{abstract}

\begin{IEEEkeywords}
Graph-based anomaly detection; Random walk; Poisoning attack; Adversarial attacks; Security and privacy.
\end{IEEEkeywords}

\section{Introduction}


Graph-based Anomaly Detection (GAD) has gained significant research attention in recent years due to the widespread use of graph data across various application domains. GAD algorithms are designed to identify anomalies in a graph, 
where nodes represent entities and edges indicate their relations. Essentially, a GAD algorithm works by initially measuring the similarities among nodes and then identifying nodes that are less similar to the rest as anomalous.
Random Walks (RWs), such as PageRank~\cite{page1999pagerank}, have emerged as a powerful tool for measuring node similarities over graphs and have become a fundamental component of many GAD systems that are extensively employed in diverse applications. 
Notably, Random-Walk-based Anomaly Detection (RWAD) has been employed in detecting money laundering within the financial industry~\cite{oliveira2021guiltywalker}, identifying fraudsters in online shopping~\cite{li2019trust}, uncovering fake accounts in social networks~\cite{yu2006sybilguard,yu2008sybillimit,shi2013sybilshield,wei2013sybildefender,jia2017random}, and serving as a general unsupervised outlier detection method for bipartite graphs~\cite{sun2005neighborhood,goodman2015using} (e.g., review data in recommender systems, stock market transaction data, and short message service), multivariate time series data~\cite{cheng2009detection,ren2017piecewise} (e.g., electrocardiograms data), and the most common feature data~\cite{yao2012anomaly,moonesinghe2006outlier,pang2016outlier,wang2018new,wang2019applying} (e.g., network intrusion detection data). Moreover, random walk has also been adopted to improve large-scale graph anomalies detection~\cite{ioannidis2021unveiling} and enhance deep-learning-based anomalies detection~\cite{zhao2020error}. These diverse applications underscore the important role of RWAD in ensuring system security.


As the accuracy of predictions produced by the RWAD methods is crucial for system security, it is essential to assess their robustness in a real-world adversarial environment. In fact, the individuals that RWAD aims to detect may have both the incentive and capability to evade detection. 
For instance, adversaries controlling bank accounts to be used in money laundering schemes may wish to remain undetected to continue their malicious activities. They could carefully manage the everyday transactions on the accounts to make them appear similar to normal ones, causing the system to falsely classify them as benign. In essence, in an adversarial environment, attackers can intentionally manipulate the input data to RWAD in order to mislead its predictions, leading to what is known as \textit{data poisoning attacks} in the literature.

\begin{figure}[tbp]
\centering
\includegraphics[width=1.0\linewidth]{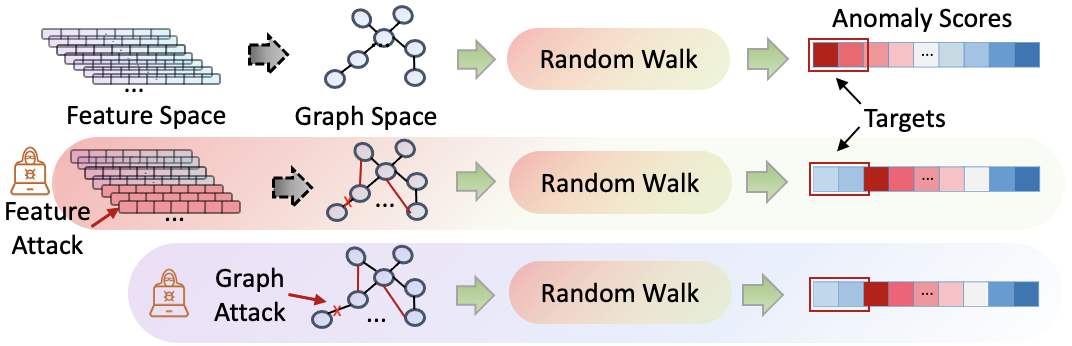}
\caption{
Illustration of RW-based anomaly detection and the distinction between attacks in the graph space and feature space.}
\label{fig:RW_AnomalyDetection}
\vspace{-10pt}
\end{figure}

However, studying the adversarial robustness of RWAD imposes new challenges due to an intriguing characteristic of RWAD. Specifically, in an RWAD system,  the graph is often not directly accessible and needs to be \textit{constructed} from raw data. 
As illustrated in Fig.~\ref{fig:RW_AnomalyDetection} (top), 
entities in the system are represented as vectors in a feature space, and a graph is then constructed based on the relationships among the entities \textit{as determined by their feature vectors}. For instance, a proximity graph can be constructed based on feature similarity. This graph is then fed into the RWAD system, which produces anomaly scores for each node in the graph. 

Consequently, there are \textit{two potential attack surfaces} against RWAD: \textbf{graph-space} attacks and \textbf{feature-space} attacks.
In graph-space attacks (Fig.~\ref{fig:RW_AnomalyDetection}, bottom), the attacker can directly modify the structure of the graph, which is a common assumption made by previous works \cite{zugner2018adversarial,z:ugner2018adversarial,zhu2022binarizedattack} that design structural attacks on graphs. In feature-space attacks (Fig.~\ref{fig:RW_AnomalyDetection}, middle), the attacker does not have direct control over the graph but can modify the features, which indirectly affects the graph's structure. It is worth noting that in the latter case, where the graph is not directly accessible,  feature-space attacks are deemed more realistic.

Unfortunately, previous research treats attacks in the graph space and feature space rather separately. On the one hand, many existing works have investigated \textit{structural attacks}~\cite{zugner2018adversarial,zhou2019attacking,zhou2020robust,zhu2022binarizedattack} against a wide range of graph learning models. On the other hand, another line of research has focused on studying feature manipulation attacks~\cite{goodfellow2014explaining,carlini2017towards,moosavi2017universal} primarily in the computer vision domain, where the data objects represented by features are independent of each other.
In contrast, one unique characteristic of RWAD is that it examines data objects that are interdependent (that is, objects are connected in a graph determined by their features), which makes the interplay between the graph-space and feature-space attacks possible. Thus, for the first time, we aim to investigate the adversarial robustness of RWAD under \textit{coupled-space} attacks, with a focus on designing more powerful realistic feature-space attacks from the guidance of graph-space attacks.

Towards this end, we begin with a formal analysis of graph-space attacks. Specifically, we define the attacks in the graph space as a decision problem. We ask whether an attacker can reduce the anomaly scores of the target nodes below a certain threshold, thereby classifying them as benign, by modifying a limited number of edges in a given graph. Our in-depth complexity analysis shows that this problem is NP-hard for both \textit{directed} and \textit{undirected} graphs. Furthermore, since feature-space attacks ultimately modify edges, they can be viewed as special cases of this problem, and the hardness results remain applicable. The hardness results serve as the anchor for us to investigate efficient attack algorithms in both the graph space and feature space.

We then proceed to design effective graph-space attacks, which are formulated as an optimization problem with the objective of minimizing the target nodes' anomaly scores output by RWAD. Solving this optimization problem encounters several challenges.
Firstly, random walk (PageRank) is an iterative algorithm that operates on an input graph, thus any changes made to the graph will require the iterations to be re-executed. Consequently, attacks against RWAD will result in a \textit{bi-level} optimization where the inner layer involves complex iteration. Second, the discrete nature of graph structure further complicates the solving of the optimization. 
To address these challenges, we propose two efficient attacks:  \textbf{alterI}-attack and \textbf{cf}-attack. The former is an iterative approach that optimizes the attack objective by projected gradient descent (PGD)~\cite{wang2019attacking} and updates the random walk model alternatively. The latter utilizes the closed form of the random walk model to transform the bi-level optimization into a single-level problem.

Finally, we investigate the more realistic feature-space attacks. Our major innovation is to use the results from the \textit{virtual} graph-space attack as our guidance to design more powerful feature-space attacks. Specifically, we utilize the guidance from two aspects: selecting the attack nodes and formulating an effective attack objective. Through extensive experiments, we demonstrate that by fully exploring the dynamics between attacks in coupled spaces, more powerful attacks could be designed, revealing more realistic security threats against RWAD systems.

The main contributions are summarized as follows:

\begin{itemize}
    \item We study the adversarial robustness of RWAD, for the first time, exploring the interplay between attacks in coupled spaces.
    \item We present a deep theoretical analysis of the hardness of attacking RWAD, which is proved to be NP-hard on both directed and undirected graphs.
    \item We propose effective attacks in coupled spaces. In particular, we innovatively utilize the results from the graph-space attacks as guidance to design more powerful feature-space attacks.
    \item We conduct comprehensive experiments to demonstrate the effectiveness of our proposed attacks. Especially we also transfer our attacks to other anomaly detection methods in the feature space. It is shown that our graph-guided feature-space attack remains effective even without knowing the target models, demonstrating a realistic threat in real-world application scenarios.
\end{itemize}

In summary, our work uncovers a unique vulnerability of RWAD and unleashes the power of attackers by exploring the interplay between attacks in coupled spaces, significantly advancing our knowledge of the adversarial robustness of RWAD in deployment.

\textbf{Road Map}: Related works (\ref{RelatedWork}) $\Rightarrow$ Target RWAD models (\ref{RW_AnomalyDetection}) $\Rightarrow$ Problem statements (\ref{ProblemSate}) $\Rightarrow$ Complexity analysis of attacks (\ref{Complexity}) $\Rightarrow$ Effective graph-space attacks (\ref{structural_attack}) $\Rightarrow$ Graph-guided feature-space attacks (\ref{FeatureAttack}) $\Rightarrow$ Evaluation (\ref{Experiments})  $\Rightarrow$ Conclusion (\ref{conclusion}).

\section{Related Works}
\label{RelatedWork}
\subsection{Graph-based anomaly detection}

This paper focus on unsupervised and node-level anomaly detection on plain and static graph. 
\textit{Random-walk-based techniques}~\cite{sun2005neighborhood,moonesinghe2006outlier,cheng2009detection,yao2012anomaly,pang2016outlier,you2017provable,wang2018new} discussed in this paper, exploiting random walk as a similarity or connectivity measurement.
Traditional \textit{feature-based} techniques \cite{akoglu2010oddball,ding2012intrusion,hooi2016fraudar} utilize statistical features, such as in and out node degrees, to extract structural information from graphs and transform the GAD to usual
anomaly detection problem. For example, OddBall ~\cite{akoglu2010oddball} built a regression model based on the density power law to estimate anomalous local patterns. These labor-intensive handcrafted features have limitations on generalizing to unknown anomalies. Beyond handcrafted features, \textit{network-representation-based} techniques, such as DeepWalk~\cite{perozzi2014deepwalk} and Node2Vec~\cite{grover2016node2vec}, are widely exploited to extract a more flexible feature representation which can be used for downstream anomaly detection tasks\cite{bandyopadhyay2020outlier}. Most recent work mainly focuses on investigating \textit{deep learning based} anomaly detection, such as DOMINANT \cite{ding2019deep} and GAL\cite{zhao2020error}.

\subsection{Adversarial attacks on graph}
Our work belongs to the category of targeted and poisoning adversarial attacks. Here, we include  the most related existing attacks on graphs. There are some previous research
efforts on the random walk (RW) based models. 
\cite{bojchevski2019adversarial} reformulate the DeepWalk model as a matrix factorization form to reduce the bi-level optimization to single-level, and then optimize the untargeted attack loss by optimizing the graph spectrum. \cite{chang2022adversarial} make further improvement to make the spectrum-based attack works in a black-box system. Different from our attacks on RW-based anomaly detection, they mainly focus on attacking
node embedding generated by RW.

In addition to RW-based model, Nettack~\cite{zugner2018adversarial}, Metattack~\cite{z:ugner2018adversarial} 
are two strong poisoning attacks for the GCN-based models. Nettack greedily selects the perturbation edges among the candidate sets with the largest gradient obtained by incremental updates. Metattack greedily selects the perturbation edges with the largest gradient obtained by meta-gradient. Note that, both of these methods can be extended to attack node features. However, Nettack does not introduce the attack node selection, and Metattack is only applicable to binary features. Furthermore, the proximity graph is different from other graphs. The proximity graph is changing along with features, while the node feature attack in \cite{zugner2018adversarial},\cite{z:ugner2018adversarial} have fixed graph structures. 
For belief propagation models, \cite{wang2019attacking} introduced a poisoning attack for graph data. For another classical graph-based anomaly detection model called OddBall, \cite{zhu2022binarizedattack} proposed BinarizedAttack which is well-designed for the binary property of edges. For graph contrastive learning, \cite{zhang2022unsupervised} attack the graph embedding by greedily choosing the most informative edges. Beyond gradient-based methods, perturbing the intrinsic property of graphs, such as spectral changes \cite{lin2022graph} shows to be more effective, but it is only suitable for untargeted attacks. These works are orthogonal to our study.

\section{Random-Walk-based Anomaly Detection}
\label{RW_AnomalyDetection}


In this section, we introduce the necessary background on unsupervised random-walk-based anomaly detection (RWAD). We first present an overview of the framework with an emphasis on the role of random walk (RW) in anomaly detection, and then give two concrete exemplar RWAD models, which are also the target models considered in this paper.


\subsection{Overview}
\subsubsection{Input data as a graph} In general, RWAD takes a \textit{plain} graph as input and produces anomaly scores for the nodes in the graph as output. In practice, the input graph could be either directly available or constructed from raw data. Depending on the levels of accessibility of the graph, we divide RWAD systems into two types: 
\begin{itemize}
    \item RWAD over \textit{directly accessible graph} (\textit{Di-RWAD}):
    In this case, the input to RWAD is a graph that represents relational data in a specific application. For instance, 
    in recommender systems, the rating towards products given be customers on E-commerce platforms can be modeled as a \textit{bipartite graph}.
    \item RWAD over \textit{indirectly accessible graph} (\textit{InDi-RWAD}): In this case, the input to RWAD are raw features of entities, and a graph is constructed as a data preprocessing step in the pipeline of anomaly detection (Fig.~\ref{fig:RW_AnomalyDetection}, top). Typically, given the feature vectors, a \textit{proximity graph} is constructed, where the nodes represent entities and an edge exists between two nodes only if they are similar enough
    in certain similarity metrics~\cite{yao2012anomaly,moonesinghe2006outlier, pang2016outlier, wang2018new}.
\end{itemize}
We note that in both cases, RWAD will operate on graphs; however, the difference lies in whether the graph is directly accessible. Later we will see that such a difference is crucial for determining the attacker's ability when designing attacks.

\subsubsection{RW as a similarity measurement}
The core of unsupervised anomaly detection is to identify data points that are significantly different from the rest of the population. RW has been shown to be an effective method for measuring the similarities of nodes in a graph. Specifically, given a graph $G=(V,E)$ with its adjacency matrix denoted as $W$, we define the transition matrix $P=(p_{ij})_{|V|\times|V|}$ as the column-normalized version of the adjacency matrix $W$, where $p_{ij}=w_{ij}/\sum_{t=1}^{|V|} w_{i,t}$. If vertex $i$ has no outgoing edges (i.e., $\sum_{t=1}^{k+n} w_{i,t}=0$), we set the transition probability to 0. The widely used Page-Rank algorithm with restart can be represented as follows:
\begin{equation}
\label{eqn:RW_general}
    \vec{s}= (1-\alpha) P \vec{s}+ \alpha \vec{r},
\end{equation}

where $\alpha$ is the restart rate, a hyper-parameter that controls the probability of restart; the vector $\vec{r}$ specifies the restart strategy, and $\vec{s}$ characterizes the node similarities. With the similarity, the anomaly score of a node is calculated as the opposite of its average similarity to all other nodes, or the average similarity among its neighbors. Next, we present two representative models to instantiate the \textit{Di-RWAD} and \textit{InDi-RWAD} systems.



\subsection{Representative Target Models}


\subsubsection{Di-RWAD}
\label{target1} 

We consider bipartite graphs as a representative example of directly accessible graphs. We next describe how to apply the RWAD algorithm to the bipartite graphs of this kind, which we term as \textit{BiGraphRW} model.

To begin, we define a bipartite graph $G=(U\cup V,E)$ as a graph with two disjoint sets of vertices $U=\{u_i|1\leq i\leq k\}$ and $V=\{v_i|1\leq i\leq n\}$, and a set of edges $E\subseteq U\times V$ that connect the vertices in $U$ to the vertices in $V$. We represent the edges in $E$ as a binary edge matrix $M=(m_{ij}){k\times n}$, where $m_{ij}=1$ if $\langle i,j\rangle\in E$, and $m_{ij}=0$ otherwise. Then, the adjacency matrix for a bipartite graph can be constructed as $W=(w_{ij})_{(k+n)\times(k+n)}=\begin{pmatrix}0 & M \\M^{T} & 0\end{pmatrix}$.  

For each node $u\in U$, \textit{BiGraphRW} applies Eqn.~\ref{eqn:RW_general} with $\vec{r}=\vec{e}_{u}$, where $\vec{e}_{u}$ is a vector with zeros element except node $u$, which means that it always restarts from node $u$. The resulting vector $\vec{s}_{u}=(1-\alpha) P \vec{s}_u+ \alpha \vec{e}_{u}$ represents the connectivity scores of node pairs $\{ \left \langle u,t\right \rangle| t \in U\cup V$\}, which quantifies the similarity between node $u$ and others. By assumption, a node $v$ tends to have a lower mean similarity score among its neighbors if it is anomalous. We denote the average neighbor similarity as $\bar{S}_{v}$:
\begin{equation}
    \bar{S}_{v}=\frac{ \sum_{i=1}^k M_{iv} \sum_{j=1, i \neq j}^k M_{jv}\vec s_{u_{i}}(u_{j}) }{\sum_{i=1}^k M_{iv} \sum_{j=1,i \neq j}^k M_{jv}},
\end{equation}
where $\vec s_{u_{i}}(u_{j})$ represent the element corresponding to node $u_{j}$ in $\vec{s_{u_{i}}}$, which is the similarity between node $u_{i}$ and $u_{j}$. Anomaly score of node $v$ is in contract to the mean similarity score $\bar{S}_{v}$, so we denoted it by 
\begin{equation}
\label{eqn:AS_bp}
    \mathcal{A}(v)=1-\bar{S}_{v}= 
    \begin{cases} 
        \text{anomaly,}  & \text{if }\mathcal{A}(v)  \geq \theta,\\
        \text{normal node,} & \text{if }\mathcal{A}(v) < \theta,
    \end{cases}
\end{equation}
where the parameter $\theta$ is a given and fixed threshold of the anomaly detection model.

\subsubsection{InDi-RWAD}
\label{target2}
\par 
A representative way to apply RWAD to non-graph data is by constructing a proximity graph. We call this variant as \textit{ProxGraphRW} model. In this approach, the input feature data is represented as $\mathbf{X}=\{\mathbf{x}_1, \mathbf{x}_2, \cdots,\mathbf{x}_n\}$, $\mathbf{x}_{i} \in \mathbb{R}^{d}$.
The first step is to construct a proximity graph according to the similarity or distance measurement between each pair of samples. To construct a proximity graph $G=\left( V,E\right)$, the vertices $V$ represent data samples $\mathbf{x}_1, \mathbf{x}_2, \cdots ,\mathbf{x}_n$, and the edges imply the similarity among vertices. This can be achieved through similarity measures, such as Euclidean distance, cosine similarity, or correlation coefficient. We denote the similarity function between $\mathbf{x}_i$ and $\mathbf{x}_j$ as $sim(\mathbf{x}_i,\mathbf{x}_j)$. Then, proximity graphs can be constructed by different rules. 
In this paper, we take $\epsilon$-Graph~\cite{moonesinghe2006outlier,yao2012anomaly} as an example, where for every data sample $\mathbf{x}_{i}$, an edge is connected to $\mathbf{x}_{j}$ if $sim\left(\mathbf{x}_{i},\mathbf{x}_{j}\right) > \epsilon$. 
We define the weighted adjacency matrix as $W=(w_{ij})_{n\times n}$, where $w_{ij}=sim(\mathbf{x}_i,\mathbf{x}_j)\cdot\mathbb{I}(sim(\mathbf{x}_i,\mathbf{x}_j)>\epsilon)$, and $\mathbb{I}(\cdot)$ is an indicator function. With the proximity graph constructed, \textit{ProxGraphRW} applies the Eqn.~\ref{eqn:RW_general} with $\vec{r}=\frac{1}{n}$, which means that the RW restart from any node with equal probability. The resulting vector $\vec{s}= (1-\alpha) P \vec{s}+\frac{\alpha}{n}$ contains the connectivity scores of all nodes, where each element $\vec{s}(v)$ quantifies the overall similarity of node $v$ to all other nodes. Finally, based on the hypothesis that anomalies have low connectivity to most others, 
the anomaly score of node $v$ is
\begin{equation}
\label{eqn:AS_px}
    \mathcal{A}(v) = 1- \vec{s}(v),
\end{equation}
where $\vec{s}(v)$ is the element corresponding to node $v$ in $\vec{s}$.

\section{Problem Statements}
\label{ProblemSate}
In this section, we introduce the adversarial environment that random-walk-based anomaly detection (RWAD) operates in, and then formally define the attack problem.

\subsection{System and Threat Model}
We consider a system consisting of two parties: an analyst who runs an RWAD algorithm to detect potential anomalies and an attacker who aims to evade the detection. In practice, the analyst would first collect data from the environment and construct a graph, which is fed into the RWAD system for anomaly detection. However, the attacker could tamper with the data collection process which will result in a poisoned graph, leading to the malfunction of the system. For instance, in online shopping platforms, the attacker may manipulate some users to provide fake ratings for target items. The resulting poisoned data can lead to biased recommendations from the recommender system.

We further introduce the threat model by specifying the attacker’s knowledge, goal, and capability. By Kerckhoffs’s principle, we assume a worst-case scenario where the attacker knows all the data as well as the anomaly detection model, which is a common assumption employed by many previous attacks~\cite{zhu2022binarizedattack,anelli2021adversarial}. 
We assume that the attacker has a set of target nodes in mind. Initially, the target nodes would have been determined as abnormal by the RWAD system if no data was manipulated. The attacker then tries to decrease the anomaly scores of those target nodes in the hope that they would evade the detection. To this end, the attacker can manipulate the data constrained by a certain budget. Specifically, depending on whether the graph is directly accessible or not, we divide the attacks into two types:
\begin{itemize}
    \item \textit{Graph-space} attack: the attacker can directly modify the structure of the graph by adding and deleting the edges under a budget constraint $K$.
    \item \textit{Feature-space} attack: the attacker can only modify the features of a set of attack nodes, which will indirectly cause changes in the graph structure. Considering a practical scenario that the targeted anomaly nodes are crafted to have specific malicious functions, we can not modify their features arbitrarily. Therefore, an indirect feature attack, aiming to decrease the anomaly scores of target nodes while keeping their features unchanged, is ideal for such a problem. Hence, we restrict the selection of attack nodes to those other than the target nodes.
\end{itemize}

\subsection{Problem definition}
To facilitate our theoretical analysis, we formally define the attacks against RWAD as follows.

\begin{definition}[\textbf{PA-RWAD}: poisoning attacks against RWAD]
\label{def:decision-problem}
An instance of the problem is defined by a tuple, $(G,\target,\anom,\threshold,\budget,\FA,\FR)$, where $G=(V,E)$ is a network, $\target \subseteq V$ is the set of targets, $\anom: \mathbb{G} \times V \rightarrow \mathbb{R}$ is the anomaly score function, $\threshold \in \mathbb{N}$ is the safety threshold, $\budget \in \mathbb{N}$ is the budget specifying the maximum number of edges that can be added or removed, $\FA \subseteq (V \times V) \setminus E$ is the set of edges that can be added, and $\FR \subseteq E$ is the set of edges that can be removed. The goal is then to identify two sets, $A^* \subseteq \FA$ and $R^* \subseteq \FR$, such that $|A^*|+|R^*| \leq K$, and for $G^*=(V, (E \cup A^*) \setminus R^*)$ we have:
\[
\left|\left\{ v_i \in V : \forall_{v_j \in \target} \anom(G^*,v_i) > \anom(G^*,v_j)  \right\}\right| \geq \threshold.
\]
\end{definition}
In practice, the top-$\threshold$ nodes ranked by their anomaly scores in descending order are determined as anomalous.
Then, the goal of \textbf{PA-RWAD} is to find a way of modifying the network by adding and removing edges, so that there are at least $\threshold$ nodes with anomaly scores greater than any of the target nodes. In other words, the target nodes are considered as benign.

We note that although \textbf{PA-RWAD} emphasizes modifying the structure of the graph, a feature-space attack is still an instance of \textbf{PA-RWAD}, since the modification of features will ultimately lead to the changes of the graph.


\section{Complexity Analysis}
\label{Complexity}


We now proceed to analyze the computational complexity of the attacks against RWAD. We summarize the hardness results in Tab.~\ref{table-hard}.
\begin{table}[!h]
\caption{Hardness results of \textbf{PA-RWAD}.}
\label{table-hard}
\centering
\begin{tabular}{|c|c|c|}
\hline
                 & Directed graph                                                          & Undirected graph                                             \\ \hline
\textbf{PA-RWAD} & \begin{tabular}[c]{@{}c@{}}NP-hard\\ (\textbf{Lemma 1} \& \textbf{Theom. 1})\end{tabular} & \begin{tabular}[c]{@{}c@{}}NP-hard\\ (\textbf{Theom. 2})\end{tabular} \\ \hline
\end{tabular}
\end{table}


\begin{theorem}
\label{thrm:decision-directed}
The \textbf{PA-RWAD}  problem is NP-hard given a directed graph.
\end{theorem}

\begin{proof}
We will prove that the problem is NP-hard by showing a reduction from the NP-complete \textit{3-Set Cover} problem. An instance of this problem is defined by a collection of subsets  $Q=\{Q_1, \ldots, Q_{|Q|}\}$ of the universe $U=\{u_1, \ldots, u_{|U|}\}=\bigcup_{Q_i \in Q}Q_i$ such that $\forall_i |Q_i|=3$, and a number $k \in \mathbb{N}$. The goal is to determine whether there exist at most $k$ elements of $Q$ that cover the entire universe, i.e., $Q^* \subseteq Q$ such that $|Q^*| \leq k$ and $U = \bigcup_{Q_i \in Q^*}Q_i$.

Let $(Q,k)$ be a given instance of the 3-Set Cover problem. We will now construct an instance of the \textbf{PA-RWAD} problem. In what follows, let $Q(u_i)$ be the subsets in $Q$ that contain $u_i$, i.e.,  $Q(u_i)=\{Q_j \in Q : u_i \in Q_j\}$. Let us also assume that $|Q| \geq 4$, as all smaller instances can be easily solved in constant time. First, we construct a directed network $G_Q=(V,E)$, where:
\begin{itemize}
\item $V = U \cup \bigcup_{Q_i \in Q} \{Q_i,q_i,o_i\} \cup \{h_1,h_2,h_3\} \cup \bigcup_{u_i \in U} \bigcup_{j=1}^{|Q|-|Q(u_i)|} \{x_{i,j},y_{i,j},z_{i,j}\},$
\vspace{3pt}
\item $E = \bigcup_{u_i \in Q_j} \{(Q_j,u_i)\} \cup \bigcup_{o_i \in V} \bigcup_{h_j \in V} \{(o_i,h_j)\} \cup \bigcup_{x_{i,j} \in V} \{(x_{i,j},u_i),(x_{i,j},y_{i,j}),(x_{i,j},z_{i,j})\}.$
\end{itemize}
An example of this construction (e.g., $|U|=5$, $|Q|=3$) is presented in Fig.~\ref{fig:decision-directed}. 
\begin{figure}
\centering
\includegraphics[width=.8\linewidth]{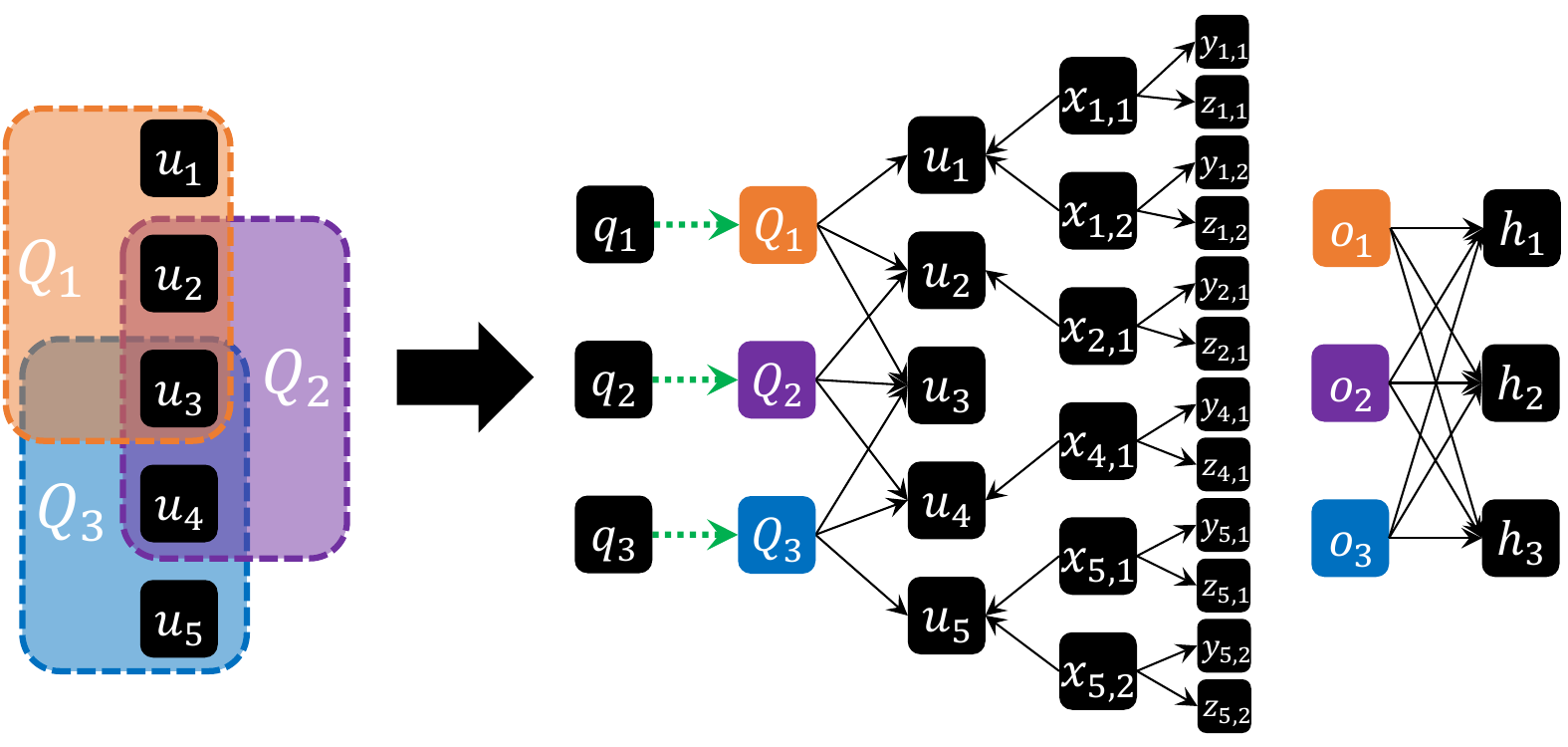}
\caption{
An example of the construction used in the proof of Theorem~\ref{thrm:decision-directed}. The green dotted arrows represent edges that can be added. 
}
\label{fig:decision-directed}
\end{figure}
Now, consider the instance $(G_Q,\target,\anom,\threshold,\budget,\FA,\FR)$ of the \textbf{PA-RWAD} problem, where:
\begin{itemize}
\item $G_Q$ is the network we just constructed,
\item $\target = U$ is the target set,
\item $\anom$ is the anomaly score function with the restart rate parameter $\alpha=\frac{1}{|Q|}$,
\item $\threshold=n-|U|$ is the safety threshold,
\item $\budget=k$ is the budget,
\item $\FA = \bigcup_{Q_i \in Q} \{(q_i,Q_i)\}$, i.e., only edges from $q_i$ to corresponding $Q_i$ can be added,
\item $\FR = \emptyset$, i.e., none of the edges can be removed.
\end{itemize}

Since $\FR = \emptyset$, for any solution to the constructed instance of the \textbf{PA-RWAD} problem, we must have $R^* = \emptyset$. Hence, we will omit the mentions of $R^*$ in the remainder of the proof, and we will assume that a solution consists just of $A^*$. We next prove a useful lemma.

\begin{lemma}
\label{lem:directed-construction}
Let $A \subseteq \bigcup_{Q_i \in Q} \{(q_i,Q_i)\}$, and let $G_Q \cup A = (V,E \cup A)$. We have that:
\[
\forall_{u_i \in U} \forall_{v \notin U} \anom(G_Q \cup A,v) > \anom(G_Q \cup A,u_i)
\]
if and only if $\forall_{u_i \in U} \exists_{(q_j,Q_j) \in A} u_i \in Q_j.$
\end{lemma}

\begin{proof}
From the formula of the anomaly score function, we have that $\anom(G_Q \cup A,v_i) = 1 - \simil(G_Q \cup A,v_i)$, where:
\[
\simil(G_Q \cup A,v_i) = \frac{\alpha}{n} + (1-\alpha) \textstyle\sum_{v_j \in V} \simil(G_Q \cup A,v_j) P_{j,i}.
\]
Therefore, we have that $\anom(G_Q \cup A,v_i) > \anom(G_Q \cup A,v_j)$ if and only if $\simil(G_Q \cup A,v_i) < \simil(G_Q \cup A,v_j)$. Let $A(u_i)$ be the set of $Q_j$ containing $u_i$ that got connected to the corresponding node $q_j$ via the edges in $A$, i.e., $A(u_i) = \{Q_j \in Q: u_i \in Q_j \land (q_j,Q_j) \in A \}$. We now compute the values of $\simil(G_Q \cup A,v_i)$ for all nodes in $V$:
\begin{itemize}
\item $\simil(G_Q \cup A,q_i) = \simil(G_Q \cup A,x_{i,j}) = \simil(G_Q \cup A,o_i) = \frac{\alpha}{n} = \frac{1}{|Q|n}$, as nodes $q_i$, $x_{i,j}$, and $o_i$ do not have any predecessors,
\item $\simil(G_Q \cup A,y_{i,j}) = \simil(G_Q \cup A,z_{i,j}) = \frac{\alpha}{n} + (1-\alpha) \simil(G_Q \cup A,x_{i,j}) \frac{1}{3} = \frac{\alpha}{n} + \frac{(1-\alpha)\alpha}{3n} = \frac{(4-\alpha)\alpha}{3n} = \frac{\left(4-\frac{1}{|Q|}\right)}{3|Q|n}$, as the only predecessor of nodes $y_{i,j}$ and $z_{i,j}$ is the node $x_{i,j}$ with out-degree $3$,
\item $\simil(G_Q \cup A,h_i) = \frac{\alpha}{n} + (1-\alpha) \sum_{o_j \in V} \simil(G_Q \cup A,o_j) \frac{1}{3} = \frac{\alpha}{n} + \frac{|Q|(1-\alpha)\alpha}{3n} = \frac{((1-\alpha)|Q|+3)\alpha}{3n} = \frac{|Q|+2}{3|Q|n}$, as the predecessors of $h_i$ are all $|Q|$ nodes $o_j$, each with out-degree $3$,
\item if $(q_i,Q_i) \notin A$ then $\simil(G_Q \cup A,Q_i) = \frac{\alpha}{n} = \frac{1}{|Q|n}$, as such node $Q_i$ has no predecessors,
\item if $(q_i,Q_i) \in A$ then $\simil(G_Q \cup A,Q_i) = \frac{\alpha}{n} + (1-\alpha)\simil(G_Q \cup A,q_i) = \frac{\alpha}{n} + (1-\alpha)\frac{\alpha}{n} = \frac{2|Q|-1}{|Q|^2 n}$, as the only predecessor of such node $Q_i$ is the node $q_i$,
\item $\simil(G_Q \cup A,u_i) = \frac{\alpha}{n} + (1-\alpha) \sum_{Q_j \in Q(u_i)} \simil(G,Q_j) \frac{1}{3} + (1-\alpha) \sum_{x_{i,j} \in V} \simil(G,x_{i,j}) \frac{1}{3} = \frac{\alpha}{n} + \frac{|Q|(1-\alpha)\alpha}{3n} + |A(u_i)|\frac{(1-\alpha)^2\alpha}{3n} = \frac{((1-\alpha)|Q|+3+|A(u_i)|(1-\alpha)^2)\alpha}{3n} = \frac{|Q|+2+|A(u_i)|\left(1-\frac{1}{|Q|}\right)^2}{3|Q|n}$, as the predecessors of $u_i$ are $|Q(u_i)|$ nodes $Q_j$, as well as $|Q|-|Q(u_i)|$ nodes $x_{i,j}$, each with out-degree $3$.
\end{itemize}

We now prove the main equivalence of the lemma. Assume that $\forall_{u_i \in U} \forall_{v \notin U} \anom(G_Q \cup A,v) > \anom(G_Q \cup A,u_i)$. In particular, it implies that: $
\forall_{u_i \in U}\, \simil(G_Q \cup A,u_i) - \simil(G_Q \cup A,h_1) > 0.
$
By substituting the values in the inequality, we get:
\[
\forall_{u_i \in U}\, \frac{|A(u_i)|\left(1-\frac{1}{|Q|}\right)^2}{3|Q|n} > 0,
\]
which in turn implies that $\forall_{u_i \in U} |A(u_i)| > 0$. Hence, we have that for every $u_i \in U$ there exists at least one $Q_j$ such that $u_i \in Q_j$ and $(q_j,Q_j) \in A$.

To prove the implication in the other direction, assume that $\forall_{u_i \in U}\, \exists_{(q_j,Q_j) \in A} u_i \in Q_j$. Hence, we get that $\forall_{u_i \in U} |A(u_i)| > 0$, which implies that: $
\forall_{u_i \in U}\, \simil(G_Q \cup A,u_i) \geq \frac{|Q|+2+\left(1-\frac{1}{|Q|}\right)^2}{3|Q|n}.
$
By comparing this value to the values computed above, we have that $\forall_{u_i \in U}, \forall_{v \notin U}$:
\[
 \simil(G_Q \cup A,v) < \frac{|Q|+2+\left(1-\frac{1}{|Q|}\right)^2}{3|Q|n} \leq \simil(G_Q \cup A,u_i),
\]
which in turn implies that:
\[
\forall_{u_i \in U} \forall_{v \notin U}\, \anom(G_Q \cup A,v) > \anom(G_Q \cup A,u_i).
\]
This concludes the proof of the lemma.
\end{proof}

Let $Q^* \subseteq Q$ be a solution to the given instance of the 3-Set Cover problem, i.e., $|Q^*| \leq k$ and $\forall_{u_i \in U} \exists_{Q_j \in Q^*} u_i \in Q_j$. From Lemma~\ref{lem:directed-construction} we have that $\anom(G_Q \cup A^*,v) > \anom(G_Q \cup A^*,u_i)$ where $A^*=\{(q_i,Q_i) : Q_i \in Q^*\}$. Hence, in network $G_Q \cup A^*$ all $\threshold=n-|U|$ nodes other than the nodes in $U$ have greater anomaly scores than all the nodes in $U$, and $|A^*| \leq k=\budget$. Therefore, $A^*$ is a solution to the constructed instance of the \textbf{PA-RWAD} problem.

To prove the implication in the other direction, assume that $A^*$ is a solution to the constructed instance of the \textbf{PA-RWAD} problem. In particular, it implies that $|A^*| \leq \budget = k$ and $\forall_{u_i \in U} \forall_{v \notin U} \anom(G_Q \cup A,v) > \anom(G_Q \cup A,u_i)$. From Lemma~\ref{lem:directed-construction} we have that $\forall_{u_i \in U} \exists_{(q_j,Q_j) \in A} u_i \in Q_j$. Therefore, $\{Q_i \in Q : (q_i,Q_i) \in A^*\}$ is a solution to the given instance of the 3-Set Cover problem.

We have shown that the constructed instance of the \textbf{PA-RWAD} problem has a solution if and only if the given instance of the 3-Set Cover problem has a solution, which concludes the proof of NP-hardness.
\end{proof}

\begin{theorem}
\label{thrm:decision-undirected}
The \textbf{PA-RWAD}  problem is NP-hard given an undirected graph.
\end{theorem}

\begin{proof}
We will prove that the problem is NP-hard by showing a reduction from the NP-complete \textit{Finding $k$-Clique} problem. An instance of this problem is defined by a network $G'=(V',E')$, and a number $k \in \mathbb{N}$. The goal is to determine whether there exist $k$ nodes that induce a clique in $G'$.

Let $(G',k)$ be a given instance of the Finding $k$-Clique problem. We will now construct an instance of the \textbf{PA-RWAD} problem. Let $n'=|V'|$, and let $d(G,v)$ be the degree of $v$ in network $G$, i.e., $d(G,v)=|\{w \in V : (v,w) \in E\}|$. First, we construct a undirected network $G=(V,E)$, where:
\begin{itemize}
\item $V = V' \cup \{t\} \cup \bigcup_{v'_i \in V'} \bigcup_{j=1}^{n'+k-d(G',v')-3} \{x_{i,j}\},$
\item $E = E' \cup \bigcup_{v'_i \in V'} \{(t,v'_i)\} \cup \bigcup_{x_{i,j} \in V} \{(v'_i,x_{i,j})\}.$
\end{itemize}
An example of this (e.g., $|V'|=4,k=3$) construction is presented in Fig.~\ref{fig:decision-undirected}.
\begin{figure}
\centering
\includegraphics[width=.6\linewidth]{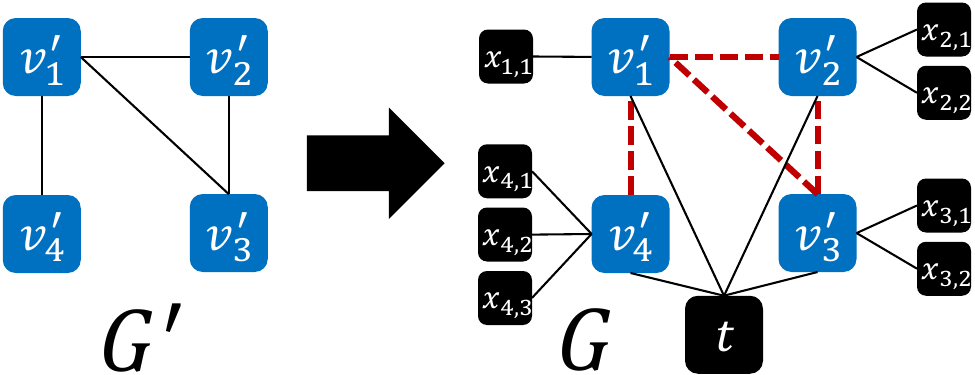}
\caption{
An example of the construction used in the proof of Theorem~\ref{thrm:decision-undirected}. The red dashed lines represent edges that can be removed.
}
\label{fig:decision-undirected}
\end{figure}
Now, consider the instance $(G,\target,\anom,\threshold,\budget,\FA,\FR)$ of the \textbf{PA-RWAD} problem, where:
\begin{itemize}
\item $G$ is the network we just constructed,
\item $\target = \{t\}$ is the target set,
\item $\anom$ is the anomaly score function with the restart rate parameter $\alpha=0$,
\item $\threshold=n-(n'-k+1)$ is the safety threshold,
\item $\budget=\frac{k(k-1)}{2}$ is the budget,
\item $\FA = \emptyset$, i.e., none of the edges can be added,
\item $\FR = E'$, i.e., only edges existing in $G'$ can be removed from $G$.
\end{itemize}

Since $\FA = \emptyset$, for any solution to the constructed instance of the \textbf{PA-RWAD} problem, we must have $A^* = \emptyset$. Hence, we will omit the mentions of $A^*$ in the remainder of the proof, and we will assume that a solution consists just of $R^*$.

From the formula of the anomaly score function with $\alpha = 0$ we have that $\anom(G,v_i) = 1 - \simil(G, v_i)$, where:
\[
\simil(G,v_i) = \textstyle\sum_{v_j \in V} \simil(G,v_j) P_{j,i}.
\]
Therefore, we have that $\anom(G,v_i) > \anom(G,v_j)$ if and only if $\simil(G,v_i) < \simil(G,v_j)$.

Moreover, from Perra and Fortunato~\cite{perra2008spectral}, we have that for the stationary distribution $\simil$ of this form (i.e., for $\alpha = 0$) in an undirected network $G$ we have that $\simil(G,v_i) \sim d(G,v_i)$, i.e., the value of the entry in $\simil$ for a given node is proportional to its degree. Therefore, we have that $\anom(G,v_i) > \anom(G,v_j)$ if and only if $d(G,v_i) < d(G,v_j)$.
Let us now compute the values of $d(G,v_i)$ for all nodes in $G$:
\begin{itemize}
\item $d(G,t) = n'$, as the node $t$ is connected with all $n'$ nodes $v'_i$,
\item $d(G,x_{i,j}) = 1 < d(G,t)$, as each node $x_{i,j}$ is only connected with the node $v'_i$,
\item $d(G,v'_i) = 1 + d(G',v'_i) + n'+k-d(G',v'_i)-3 = n'+k-2 \geq d(G,t)$, as each node $v'_i$ is connected with the node $t$, $d(G',v'_i)$ nodes from $V'$, as well as $n'+k-d(G',v'_i)-3$ nodes $x_{i,j}$.
\end{itemize}

Since $\threshold=n-(n'-k+1)$, all nodes $x_{i,j}$ have a smaller degree than $t$, and the total number of $x_{i,j}$ is $n-n'-1$, we need at least $k$ out of $n'$ nodes in $V'$ to have a smaller degree than $t$ in order for the safety threshold to be satisfied. However, they all have equal or greater degrees than $t$. Hence, the safety threshold is not satisfied in $G$.

Since the removal of edges from $\FR$ can only change the degrees of nodes in $V'$, we need to decrease the degree of $k$ of these nodes to a value smaller than that of $t$. For each of these $k$ nodes we have to remove at least $\Delta$ edges incident with it, where:
\[
\Delta = d(G',t) - d(G',v'_i) + 1 = n'+k-2 - n' + 1 = k-1.
\]

Let $V^* \subseteq V'$ be a solution to the given instance of the Finding $k$-Clique problem, i.e., a set of $k$ nodes forming a clique in $G'$. Since $\FR=E'$ and the degree of each node in $k$-clique is $k-1$, we have that $V^* \times V^* \subseteq \FR$, and removing $V^* \times V^*$ from $G$ decreases the degree of $k$ nodes from $V'$ by $\Delta=k-1$ each. Therefore, $V^* \times V^*$ is a solution to the constructed instance of the \textbf{PA-RWAD} problem.

To prove the implication in the other direction, assume that $R^*$ is a solution to the constructed instance of the \textbf{PA-RWAD} problem. At least $\frac{k\Delta}{2}=\frac{k(k-1)}{2}$ of the removed edges have to be incident with the $k$ nodes from $V'$ contributing to the safety threshold. However, since the total budget is $\budget=\frac{k(k-1)}{2}$, all of the removed edges have to be incident with the $k$ nodes from $V'$ contributing to the safety threshold, and $\frac{k(k-1)}{2}$ edges incident with $k$ nodes constitute a clique. Since we have that $\FR=E'$, the same edges constitute a $k$-clique in $G'$. Therefore, $\bigcup_{(v'_i,v'_j) \in R^*} \{v'_i, v'_j\}$ is a solution to the given instance of the Finding $k$-Clique problem.

We have shown that the constructed instance of the \textbf{PA-RWAD} problem has a solution if and only if the given instance of the Finding $k$-Clique problem has a solution, which concludes the proof of NP-hardness.
\end{proof}

\section{Practical Graph-Space Attacks}
\label{structural_attack}
In this section, we investigate practical attacks in the graph space. We note that the graph-space attack itself is important in the case where the graph is directly accessible. Moreover, as we will show later, the results of graph-space attacks provide insightful guidance for feature-space attacks.

\subsection{Attack Formulation}
We begin by formulating the decision problem \textbf{PA-RWAD} as an optimization problem.
We use $G =(V,E)$ with its corresponding adjacency matrix $W$ to represent the original clean graph. We assume that the anomaly detection system predicts node $v$ as an anomaly if the anomaly score $\mathcal{A}(v)$ is greater than a threshold $\theta$. The attacker aims to decrease the number of nodes in a given target set $\mathcal{T} \subset V$ that are identified as anomalies by modifying at most $K$ edges in the graph. To represent the edge manipulations, we denote the modification by a binary matrix $B=(b_{uv})_{(|V|\times|V|)}$, where the element $b_{uv} \in \{0,1\}$. If $b_{uv}=0$, the edge $\langle u,v \rangle$ remains unchanged, and $b_{uv}=1$ lead to add/delete of edge $\langle u,v \rangle$. Then the attack graph can be represented by $|W-B|$. In this paper, we consider undirected graphs where the adjacency matrix is always symmetric, and the budget constraint can be represented as $\sum_{u>v} b_{uv} \le K$. Then the graph-space attack problem can be formulated as follows: 
\begin{equation}
\label{opt:graphAttack}
    \begin{array}{rlclcl}
        \displaystyle \min\limits_{B} & \multicolumn{3}{l}{\sum_{v\in \mathcal{T}} \mathbb{I}(\mathcal{A}(v)>\theta),} \\
        \textrm{s.t.} & b_{uv}  \in  \{0,1\}, \, \sum_{u>v} b_{uv} \le  K,
    \end{array}
\end{equation}
where $\mathbb{I}(\cdot)$ is a indicator function, $\mathbb{I}(\mathcal{A}(v)>\theta)=1$ if the anomaly scores of node $v$ is greater than $\theta$.

\subsection{Attack Method}

\par To address the non-differentiable issue of the binary values in $B$, we adopt a relaxation strategy by representing $b_{uv}$ in a continuous space that ranges from 0 to 1. This is denoted as $\tilde{B}$, which is subsequently converted back to binary form $\bar{B}$ after solving the optimization problem. 
To handle the discrete objective function in Eqn.~\ref{opt:graphAttack}, we replace it with the sum of anomaly scores among target nodes, $\mathcal{L}_a(\tilde{B})= \sum_{v \in \mathcal{T}} \mathcal{A}(v)$, then we can re-formulate the attack problem as:
\begin{equation}
\label{opt:graphAttack_relaxed}
    \begin{array}{rlclcl}
        \displaystyle \min\limits_{\tilde{B}} & \multicolumn{3}{l}{ \mathcal{L}_a(\tilde{B})= \sum_{v \in \mathcal{T}} \mathcal{A}(v),} \\
        \textrm{s.t.} & \tilde{b}_{uv} \in [0,1],\, \sum_{u>v} \bar{b}_{uv} \le K, \\
    \end{array}
\end{equation}
where $\tilde{B}$ is the relaxed and continuous adjacency matrix, $\bar{B}=(\bar{b}_{ij})$ is the discrete version of $\tilde{B}$.  

To solve the challenging bi-level optimization problem, we propose two strategies: alternative iteration attack (\textbf{alterI}-attack) and closed-form attack (\textbf{cf}-attack). 
In brief, the \textbf{alterI}-attack iterates the inner RW model and the attack optimization alternatively to approximate the bi-level optimization, while the \textbf{cf}-attack transforms the bi-level optimization into a single-level problem. We first introduce the \textbf{alterI}-attack and then highlight the difference in the \textbf{cf}-attack.


\subsubsection{\textbf{alterI}-attack}
The optimization of problem~\eqref{opt:graphAttack_relaxed} remains a challenging task due to the need to reverse the continuous variable $\bar{B}$ to binary $\bar{B}$ while satisfying the budget constraint. To overcome this difficulty, we first use projected gradient descent (PGD) to efficiently optimize $\tilde{B}$ without considering the budget constraint $\sum_{u>v} \bar{b}_{uv} \le K$. Then, we obtain the binary matrix $\bar{B}$ by selecting the top-$K$ elements from $\tilde{B}$. This approach allows us to efficiently approximate the constrained optimization problem while ensuring that the attack budget is satisfied. However, optimizing the relaxed optimization problem is still challenging because the anomaly score $\mathcal{A}(v)$ in the loss function $\mathcal{L}_a(\tilde{B})$ depends on the variable $\tilde{B}$ in a complex way. After updating $\tilde{B}$, obtaining $\mathcal{A}(v)$ requires iterating over Eqn.~\ref{eqn:RW_general} dozens of times to get the converged node similarity vector $\vec{s}$, and the gradient needs to be traced back to each iteration. To address this issue, we only iterate over Eqn.~\ref{eqn:RW_general} once instead of multiple times.
The detailed procedures are summarized in Alg.~\ref{alg:graphAttack} and Fig.~\ref{fig:Attack_Framework} (top). 
Firstly, we update the adjacency matrix with $\tilde{W}=|W-\tilde{B}|$ (line:5), and then we update the similarity score $\mathcal{A}(v)$ based on $\tilde{W}$ for one step using Eqn.~\ref{eqn:RW_general} and then obtain the anomaly score with Eqn.~\ref{eqn:AS_bp} or \ref{eqn:AS_px} (line:6-9). Next, we update attack loss $\mathcal{L}_a (\tilde{B})$ based on $\mathcal{A}(v)$ (line:10), and calculate the projected gradient to optimize $\tilde{B}$ for one step (line: 11-15). Repeating the alternative iteration leads to the convergence of the inner model $\vec{s}$ and also the continuous attack variable $\tilde{B}$. After the iterations, we keep the top-$K$ elements in $\tilde{B}$ to obtain $\bar{B}$ and the others are set to zeros (line:17-18). Finally, the attacked graph is obtained by $\hat{W}=|W-\bar{B}|$ (line:19). This algorithm is also suitable for weighted graphs in which the weights on edges are in $[0,1]$, and the final solution is to modify $K$ edge weights while the other weights remain unchanged. 
\begin{algorithm}[htb]
\caption{Graph-space attack.}
\label{alg:graphAttack}
\begin{algorithmic}[1]
\State \textbf{Input:} Graph with adjacency matrix $W$, attack budget $K$, attack iteration $T$, learning rate $\eta$.
\State \textbf{Output:} Attacked graph with  adjacency matrix $\hat{W}$.
\Function{AlterI-attack}{$W$, $K$, $T$, $\eta$}
\For {$t=1$ to $T$}
    \State Update adjacency matrix: $\tilde{W}=|W-\tilde{B}|$.
    \For{each node $v$ in target set $\mathcal{T}$}
    \State Update similarity scores $\vec{s}$ with Eqn.~\ref{eqn:RW_general}.
    \State Update anomaly score $\mathcal{A}(v)$ based on $\vec{s}$.
    \EndFor
    
    \State Update objective function $\mathcal{L}_a (\tilde{B})$ with $\mathcal{A}(v)$.
    \For {each edge $\tilde{b}_{uv}$ in $\tilde{B}$}
        \State Calculate gradient $g_{uv}=\tilde{b}_{uv}-\eta \frac{\partial \mathcal{L}(\tilde{B}^)}{ \partial \tilde{b}_{uv}}$
        \State Project $g_{uv}$ into $[0,1]$
        \State Update $\tilde{b}_{uv}$ in $\tilde{B}$
    \EndFor 
    \EndFor
\State Choose top-$K$ edges in $\tilde{B}$ to obtain $\bar{B}$:
\State
    \begin{equation}
    \bar b_{uv}= 
    \begin{cases} 
        \tilde{b}_{uv}  & \text{if }\tilde{b}_{uv} \in  top_K(\tilde{B}),\nonumber\\
        0 & \text{otherwise}.
    \end{cases}
    \end{equation}
\State Obtain attacked graph $\hat{W}=|W-\bar{B}|$.
\State \Return $\hat{W}$.
\EndFunction
\end{algorithmic}
\end{algorithm}
\subsubsection{\textbf{cf}-attack}
While the \textbf{alterI}-attack approach is feasible, the one-step update of the inner model is a simple estimation that may not provide accurate attack loss during the iteration. To address this issue and obtain accurate attack loss, we employ the closed-form solution of the inner model to transform the bi-level optimization problem into a single-level problem. 
According to \cite{boldi2007deeper,gasteiger2018combining}, the inner model (Eqn.~\ref{eqn:RW_general}) has closed-form solution as follows:
\begin{equation}
    \label{eqn:RW_closed-form}
    \vec{s} = \alpha (I- (1-\alpha)P)^{-1}\vec{r},
\end{equation}
where $I$ is an identity matrix. With the closed-form solution, we can directly obtain the accurate anomaly scores after the update of $\tilde{B}$. In contrast to the \textbf{alterI}-attack, which iterates the inner model once after updating $\tilde{B}$, our innovative \textbf{cf}-attack approach substitutes the Eqn.~\ref{eqn:RW_general} (line:7) with Eqn.~\ref{eqn:RW_closed-form} to  
obtain the accurate connectivity scores $\vec{s}$ for current $\tilde{B}$, and others remain the same. While \textbf{cf}-attack offers a more accurate formulation than \textbf{alterI}-attack, it comes with the cost of potential time consumption when calculating the matrix inverse, particularly for graphs with a large number of nodes or edges. In contrast, \textbf{alterI}-attack does not encounter such a problem, making it a more efficient option for such scenarios. Both \textbf{cf}-attack and \textbf{alterI}-attack have their own unique advantages. 

\begin{figure}[tbp]
\centering
\includegraphics[width=.9\linewidth]{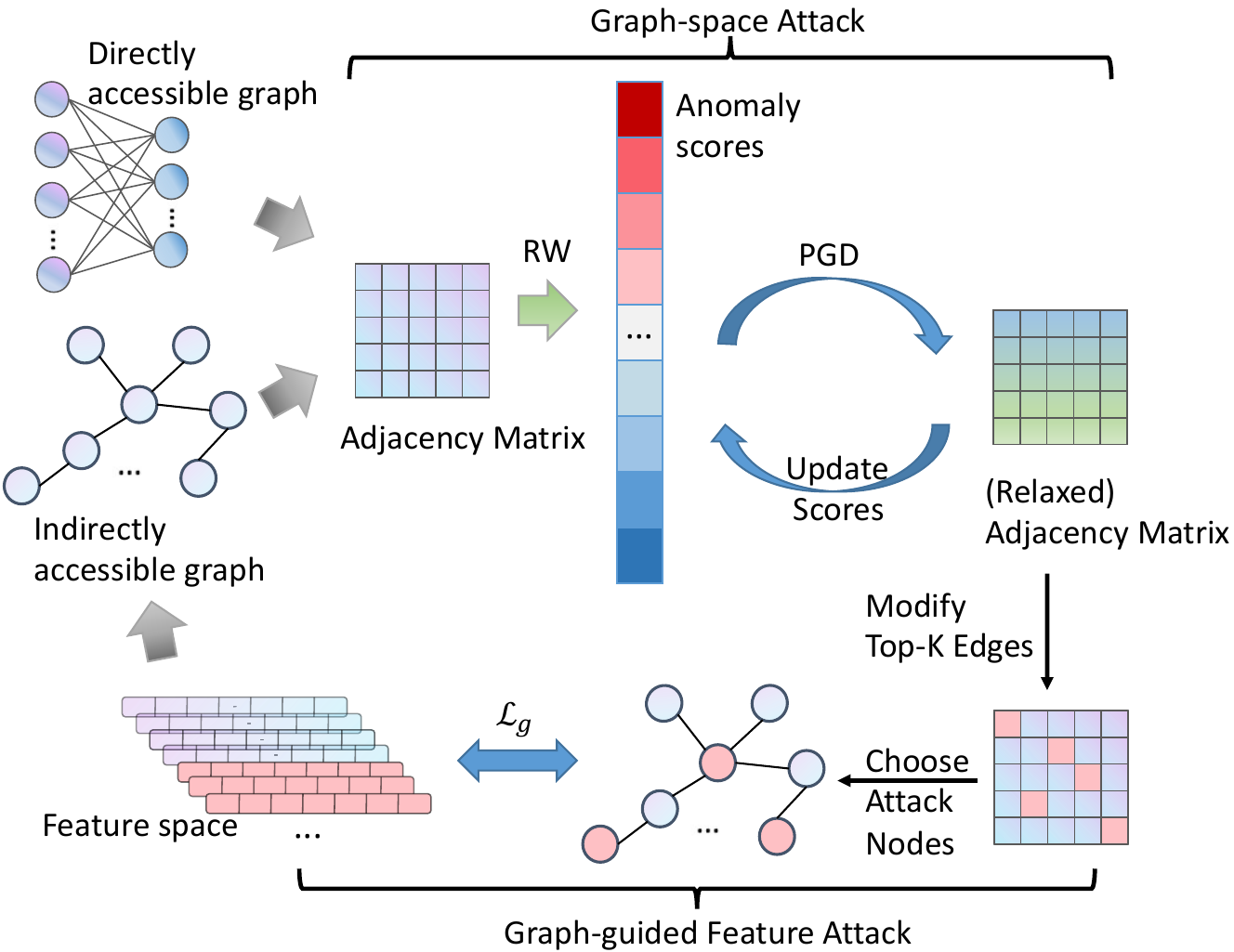}
\caption{
Illustration of proposed attacks.}
\label{fig:Attack_Framework}
\end{figure}

\section{Graph-Guided Feature-Space Attacks}
\label{FeatureAttack}


\subsection{Motivation for Feature-space Attacks}
Previously, we presented effective graph-space attacks against Di-RWAD. However, for InDi-RWAD, where the graphs are not directly accessible, the \emph{realizability} of the attacks becomes a serious concern: the attacker cannot directly modify the edges in a virtual graph space. Instead, in many practical application scenarios, what the attacker can modify are the attributes associated with the entities in their control. 
For example, when it comes to network intrusion detection, each TCP connection represents an entity or node, and attackers can manipulate certain TCP connections by altering attributes such as connection duration, protocol type, and the number of urgent packets. Such manipulations will change the structure of the proximity graph in the \textit{ProxGraphRW} model to become perturbed, which can help shield the targeted anomaly TCP connection from being detected.

Thus, investigating feature-space attacks against InDi-RWAD is of significant practical importance. In particular, we consider the scenario where an attacker can manipulate a set of entities (corresponding to nodes in the constructed proximity graph) and modify their features to assist a group of target nodes in avoiding detection. We explore the connection between graph-space and feature-space attacks and demonstrate how guidance from graph-space attacks can be leveraged to construct effective feature-space attacks.

\subsection{Attack Formulation}
Consider a set of entities with features $\mathbf{X}=\{\mathbf{x}_1, \mathbf{x}_2, \cdots, \mathbf{x}_n\}$, where $\mathbf{x}_i \in \mathcal{X}$ denotes the feature vector associated with entity $i$. As introduced in Section~\ref{target2}, a proximity graph can be constructed from $\mathbf{X}$, where the nodes represent those entities and edges indicate similar node pairs.  An attacker aims to allow a set of target entities (nodes) $\mathcal{T}$ to evade detection. We assume that the attacker has control of a set of \textit{attack nodes} $\mathcal{Z}$ such that the features of the nodes in $\mathcal{Z}$ can be arbitrarily modified in a certain domain $\mathcal{X}$. To limit the attacker's ability, we make the restriction that $\mathcal{Z} \cap \mathcal{T} = \emptyset$ and $|\mathcal{Z}| \leq K'$. For an attack node $i \in \mathcal{Z}$, we denote the modified feature vector as $\hat{\mathbf{x}}_i$. The manipulated feature matrix is $\hat{\mathbf{X}}$.
We note that since the manipulation of the features leads to the change of graph structure, the anomaly score function $\mathcal{A}(v;\hat{\mathbf{X}})$ depends on the features $\hat{\mathbf{X}}$.
Then, we can formulate the feature-space attack as follows: 


\begin{equation}
\label{opt:feature}
    \begin{array}{rlclcl}
        \displaystyle \min\limits_{\hat{\mathbf{x}}_i, i \notin \mathcal{T}} & \multicolumn{3}{l}{\mathcal{L}(\hat{\mathbf{X}})=\sum_{v\in \mathcal{T}} \mathbb{I}(\mathcal{A}(v;\hat{\mathbf{X}})>\theta),} \\
        \textrm{s.t.}
        &  \hat{\mathbf{x}}_v = \mathbf{x}_v, \forall v\in \mathcal{T},\, \hat{\mathbf{x}}_{i} \in \mathcal{X}, \\
        & \mathcal{Z}=\{i|\hat{\mathbf{x}}_{i} \neq \mathbf{x}_{i}\},\, |\mathcal{Z}| \leq K'.
    \end{array}
\end{equation}



\subsection{Two Levels of Guidance from Graph-Space Attacks}
Applying the gradient-descent method to solve problem~\eqref{opt:feature} faces a crucial challenge: while gradient descent can be used to effectively optimize the node features, it is hard to control which nodes are to be manipulated. In other words, it is nontrivial to guarantee the constraint $|\mathcal{Z}| \leq K'$ while preserving optimization performance. We adopt a divide-and-conquer strategy to tackle this problem: we first select up to $K'$ nodes as the attack nodes and then utilize gradient descent to optimize node features.  In particular, we show that the results from graph-space attacks can be innovatively utilized to guide both the selection of attack nodes and feature optimization.

Specifically, given a proximity graph $\mathcal{G}$, we can leverage the attacks in Section~\ref{structural_attack} to produce a poisoned graph $\mathcal{G}'$. Even though $\mathcal{G}'$ cannot be directly realized, it represents an excellent candidate in the graph space with which the target nodes $\mathcal{T}$ could evade detection with high probability. Thus, our intuition is to manipulate features so that the resulting proximity graph would approximate $\mathcal{G}'$. To this end, we utilize the guidance from the following two aspects.


\paragraph{Guidance on attack node selection} In the graph-space attack, we denote the set of edges/non-edges modified by the attacker as $\mathcal{E}_a$. Intuitively, the modification of $\mathcal{E}_a$ will influence the anomaly scores of the targets most. To preserve such an influence, we set the attack nodes as those ones incident to the edges/non-edges in  $\mathcal{E}_a$. Note that we can always easily adjust the budget in the graph-space attack such that the constraint $|\mathcal{Z}| \leq K'$ is satisfied. 

After fixing the attack nodes $\mathcal{Z}$, we can follow a similar approach in the graph-space attack to optimize the features. Specifically, we replace the indicator function in~\eqref{opt:feature} with the sum of anomaly scores of target nodes. For discrete features, we relaxed their discrete feature domain to the continuous space denoted by $\tilde{\mathcal{X}}$. Then, let $\tilde{\mathbf{x}}_i\in \tilde{\mathcal{X}}$ denote the relaxed feature, and $\tilde{\mathbf{X}}=\{\tilde{\mathbf{x}}_i| i\in V\}$, the feature-space attack can be formulated as the following optimization problem:
\begin{equation}
\label{loss:anomaly}
\min\limits_{\hat{\mathbf{x}}_i\in\tilde{\mathcal{X}}, i \in \mathcal{Z}} \mathcal{L}_a(\tilde{\mathbf{X}})=\textstyle\sum_{v \in \mathcal{T}}\mathcal{A}(v;\tilde{\mathbf{X}}). 
\end{equation}

We term this type (with objective function $\mathcal{L}_a$) of feature-space attacks as \textbf{G-Guided}. We can straightforwardly adopt the two algorithms \textbf{alterI-attack} and \textbf{cf-attack} to solve the optimization problem~\eqref{loss:anomaly}, resulting in two variants named \textbf{G-Guided-alterI} and \textbf{G-Guided-cf}.

\paragraph{Guidance on reformulation of attack objective}
Beyond the selection of attack nodes, the poisoned graph $\mathcal{G}'$ obtained from the graph-space attack can provide vital information for optimizing the features. Specifically, we aim to optimize the features such that the proximity graph constructed from the modified features would approximate $\mathcal{G}'$ as much as possible. 
To this end, we reformulate the attack objective function as follows: 
\begin{equation}
    \label{loss:attack_graph}
    \mathcal{L}_g(\tilde{\mathbf{X}})=\sum_{\{(i,j)|\bar{b}_{ij}>0\}\atop i/j \in \mathcal{Z}} |sim(\mathbf{x}_i,\mathbf{x}_j)-\hat{w}_{i,j}|,
\end{equation}
where $\hat{w}_{ij}$ is the element in the attacked adjacency matrix $\hat{W}$.
This objective function aims to push the similarity between control nodes $\mathbf{x}_i$ and other nodes $\mathbf{x}_j$ (denoted by $sim(\mathbf{x}_i,\mathbf{x}_j)$) close to the manipulated edges $\hat{w}_{i,j}$ in the poisoned graph $\mathcal{G}'$. Intuitively, minimizing $\mathcal{L}_g$ allows us to approximate an inverse problem: given $\mathcal{G}'$, find the node features from which $\mathcal{G}'$ can be constructed. Since~\eqref{loss:attack_graph} is a single-level function, we can directly adopt PGD (similar to the graph-space attack) to solve the optimization problem. We term this type (with objective function $\mathcal{L}_g$) of feature-space attacks as \textbf{G-Guided-plus}. The attack algorithm in the feature space is summarized in Alg.~\ref{alg:featureAttack} and Fig.~\ref{fig:Attack_Framework} (bottom).


\begin{algorithm}[htb]
\caption{Feature-space attack.}
\label{alg:featureAttack}
\begin{algorithmic}[1]
\State \textbf{Input:} Feature matrix $\tilde{\mathbf{X}}$, attack nodes $\mathcal{Z}$, attack iteration $T$, learning rate $\eta$.
\State \textbf{Output:} Attacked feature matrix $\hat{\mathbf{X}}$.
\Function{FeatureAttack}{$\tilde{\mathbf{X}}$, $\mathcal{Z}$, $T$, $\eta$}
\For {$t=1$ to $T$}
\State Construct graph based on $\tilde{\mathbf{X}}$ (Section~\ref{target2}).
\State Update similarity scores $\vec{s}$ with Eqn.~\ref{eqn:RW_general}.
\State Update the anomaly scores based on $\vec{s}$ (Eqn.~\ref{eqn:AS_px}). 
\State Update objective function $\mathcal{L} (\tilde{\mathbf{X}})$.
\For {each attack nodes ${\tilde{\mathbf{x}}_i, i\in \mathcal{Z}}$}
\State Calculate gradient $g_{i}=\tilde{\mathbf{x}}_{i}-\eta\frac{\partial \mathcal{L}(\tilde{\mathbf{X}})}{\partial \tilde{\mathbf{x}}_{i}}$.
\State Project $g_{i,j}$ into the feasible set $\tilde{\mathcal{X}}$.
\State Update $\tilde{\mathbf{x}}_{i,j}$ in $\tilde{\mathbf{X}}$.
\EndFor
\EndFor
\State Rounding the attacked feature:
    \begin{equation}
    \nonumber
    \hat{\mathbf{x}}_i= 
    \begin{cases} 
        round(\tilde{\mathbf{x}}_{ij})  & \text{if feature $j$ is discrete},\\
        \tilde{\mathbf{x}}_{ij} & \text{otherwise}.
    \end{cases}
    \end{equation}
\Return Attacked feature matrix $\hat{\mathbf{X}}$.
\EndFunction
\end{algorithmic}
\end{algorithm}

\section{Experiments}
\label{Experiments}

In this section, we evaluate the performances of our proposed attacks
by answering these four major questions:
1) Are our proposed graph-space attacks effective? 2) What are the preferences of the proposed graph attack?
3) How effective are the graph-guided feature-space attacks? 4) How is the transferability of the graph-guided feature-space attacks?
\subsection{Datasets and Experiment Settings}
We consider four datasets that are commonly used for graph-based anomaly detection: Paper-Author, Magazine, 
KDD-99, and MINIST (outlier). Among them, the first two are bipartite graphs while the latter two datasets are feature data. Below is the detailed description. \textbf{All datasets, source code for our proposed attacks, and evaluated baselines are in our \textit{GitHub} link.}~\footnote{https://github.com/Yuni-Lai/CoupledAttackRW.}
\begin{itemize}
    \item Paper-Author~\cite{sun2005neighborhood}: This dataset contains papers crawled from the arXiv preprint database. Nodes $U$ represent papers, while nodes $V$ represent authors. An edge $\langle u, v\rangle$ indicates that the author $v$ is shown in the paper $u$. We randomly sampled $10,000$ records and deleted nodes with degrees lower than $5$, resulting in $|U|,|V|=2311,\,405$. We manually inject $10\%$ of anomaly nodes following~\cite{hooi2016fraudar}.
    \item Magazine: This dataset contains Amazon Aeviews Data~\footnote{https://nijianmo.github.io/amazon/, accessed May 2023.} under the category of Magazine Subscriptions. We randomly sampled $100,000$ records and removed nodes with degrees lower than $3$, resulting in $|U|,|V|=1079,\, 1180$ nodes. We also injected $10\%$ of anomaly nodes manually following~\cite{hooi2016fraudar}.
    \item KDD-99~\cite{moonesinghe2006outlier}: The dataset contains network intrusion data with $41$ features and $4$ types of attacks. We randomly sampled $10,000$ benign data and $100$ anomaly data for the experiment. 
    \item MINIST (outlier): This is a subset of the MINST handwritten digits dataset, created for the outlier detection task in Outlier Detection DataSets~\footnote{http://odds.cs.stonybrook.edu/, accessed May 2023.}. It contains a total of $7603$ images, with $6903$ images of digit-$0$ regarded as normal points and $700$ images of digit-$6$ regarded as outliers. Each sample has $100$ features. 
\end{itemize}

\subsection{Experimental Settings}

We conduct our experiments on Ubuntu $20.04$ system with an NVIDIA GeForce RTX $3090$ GPU, Python $3.7$, and PyTorch $1.10.0$. All the experiments are repeated $10$ times with different random seeds, and different target nodes are sampled. 
\subsubsection{Target nodes and budgets}
For attacking \textit{BiGraphRW} model, we sample $5$ target nodes from the top $100$ anomaly nodes, while in \textit{ProxGraphRW} model, we sample $20$ target nodes from the top $100$ anomaly nodes. We set the attack edge budget proportion to the sum of \textit{target node degrees} (e.g., budget $10\%: K = 0.1 \times \sum_{v\in\mathcal{T}} d(v)$, where $d(v)$ is the degree of node $v$). Setting the budget associated with node degree is commonly adopted in targeted attacks such as Nettack~\cite{zugner2018adversarial,bojchevski2019adversarial}. In feature-space attacks, we set the number of attack nodes as the number of nodes involved in the \textbf{alterI} graph-space attack.

\subsubsection{Evaluation metrics}
Our main focus is to evaluate the effectiveness of our proposed method facilitating target nodes to evade detection under different detection thresholds. Usually, the detection threshold $\theta$ is set to the proportion of data size, and we evaluate the level of detect ratio as the top $5\%$ and $10\%$. We then use the \textit{evasion rate} $\mathsf{ER}$ of target nodes under these detection thresholds as the main metric. Specifically, the evasion rate is computed as $\mathsf{ER} = n_0/|\mathcal{T}|$, where $n_0$ is the number of target nodes not shown in the top $5\%$ or $10\%$ anomaly scores (i.e., evaded successfully). Besides, we also evaluate the average anomaly scores of target nodes.
\subsubsection{Baselines}
We evaluate the effectiveness of our proposed attacks against several baselines for both graph-space attacks and feature-space attacks.
\paragraph{Graph-space attack}
The most relevant prior work is \cite{bojchevski2019adversarial}. Although this work also proposes a targeted attack for the RW model, it is specific to the DeepWalk model and cannot be directly applied to our RWAD systems. Therefore, we transfer its targeted attack to our model. Besides, we also adopt two common baselines RndAdd and DegAdd following \cite{bojchevski2019adversarial}.
\begin{itemize}
    \item \textbf{RndAdd}: This baseline randomly adds candidate edges, where the candidate edges are the edges incident to target nodes. 
    \item \textbf{DegAdd}: This baseline adds candidate edges with the top-$K$ highest degrees, where the candidate edges are also the edges incident to target nodes. 
    \item \textbf{DeepWalk}\cite{bojchevski2019adversarial}: In this baseline, we transfer the attack designed for DeepWalk to RWAD models. 
    \item Our methods: \textbf{alterI} and \textbf{cf} are our proposed attacks with alternative iteration and closed-form solution, respectively.
\end{itemize}
\paragraph{Feature-space attack}
To evaluate the effectiveness of our graph-guided attack in node selection, we include random selection as a baseline for comparison. 
\begin{itemize}
    \item \textbf{VanillaOpt}: This baseline randomly selects attack nodes from candidates and optimizes node features with the objective function $\mathcal{L}_a(\tilde{\mathbf{X}})$ in \eqref{loss:anomaly} with strategy \textbf{alterI}. 
    \item Our methods: We use the graph-space attacks to guide the selection of attack nodes and choose  $\mathcal{L}_a(\tilde{\mathbf{X}})$ as the attack objective function, resulting in two attack methods \textbf{G-guided-alterI} and \textbf{G-guided-cf}, which adopt \textbf{alterI} and \textbf{cf} to optimize node features respectively. In addition, when the objective function $\mathcal{L}_g(\tilde{\mathbf{X}})$~\eqref{loss:attack_graph} is selected, the attack method is \textbf{G-guided-plus}.
\end{itemize}
\subsubsection{Hyper-parameters}
Grid search is employed to find the optimal hyper-parameters in all the attack methods over different datasets. For \textit{BiGraphRW} model, the regularization parameter $\lambda= 1 \times 10^{-6}$, learning rate $lr=1.0$, $60$ epochs with SGD optimizer. 
For \textit{ProxGraphRW} model, we evaluate proximity graphs constructed with cosine similarity and correlation similarity. The similarity threshold $\epsilon$ for constructing the graph is $0.8$ for the KDD-99 dataset and $0.5$ for the MNIST dataset; the regularization parameter $\lambda= 1 \times 10^{-4}$, learning rate $lr=1.0$, $35$ epochs for the KDD-99 dataset and $100$ for the MNIST dataset with Adam optimizer in the graph-space attack. In feature-space attack, learning rate $lr=1.0$, $500$ epochs with Adam optimizer.

\subsection{Performance of Graph-space Attacks}

To begin with, we evaluate the performance of the target RWAD models over corresponding datasets. As shown in Tab.~\ref{tab:auc}, both models achieved an AUC (area under reception curve) of at least $0.89$, demonstrating a strong ability to identify anomalies.

\subsubsection{Effectiveness of attacks}
We present the evasion rates $\mathsf{ER}$ of those attack methods under different detection levels (top-$5\%/10\%$) in Tab.~\ref{tab:G-attack_bi} and \ref{tab:G-attack_pro}.
We observe that our proposed graph attack methods, \textbf{alterI} and \textbf{cf}, significantly outperform other baselines on all datasets. For instance, at the detection level of top-$5\%$, our results indicate that our proposed attack on \textit{BiGraphRW} model is highly effective, achieving an evasion rate of over $85\%$ with a budget of $40.0\%$. Similarly, for \textit{ProxGraphRW} model, with a budget of $60.0\%$, the evasion rate (under detection threshold top-$10\%$) is over $80\%$ on the MNIST dataset. 
Since the MNIST dataset is relatively easier to attack, we report the attack performance at a higher detection threshold. The reason why the \textbf{DeepWalk} method does not exhibit a strong attack effect could be attributed to its transferability across different types of random walk models.
\vspace{-5pt}
\begin{table}[htb]
\caption{AUC of RWAD.}
\label{tab:auc}
\centering
\begin{tabular}{lllll}
\hline
Models  & \multicolumn{2}{c}{\textit{BiGraphRW}} & \multicolumn{2}{c}{\textit{ProxGraphRW}} \\ \hline
Dataset & Author-Paper    & Magazine    & KDD-99          & MNIST         \\ \hline
AUC     & 1.00            & 0.89        & 0.98            & 0.90          \\ \hline
\end{tabular}
\vspace{-10pt}
\end{table}

\begin{table}[thb]
\centering
\setlength{\tabcolsep}{3.1pt}
\caption{Graph attack results on \textit{BiGraphRW} model.}
\label{tab:G-attack_bi}
\begin{tabular}{ccllllll}
\hline
\multicolumn{1}{l}{Dataset} & \multicolumn{1}{l}{Metrics} & budget & RndAdd & DegAdd & DeepWalk & alterI & cf \\ \hline
\multirow{12}{*}{\begin{tabular}[c]{@{}c@{}}Author-\\ Paper\end{tabular}} & \multirow{6}{*}{\begin{tabular}[c]{@{}c@{}}$\mathsf{ER}$ \\ (5\%)\end{tabular}} & 0\% & 0.560 & 0.560 & 0.560 & 0.560 & 0.560 \\
 &  & 20\% & 0.560 & 0.560 & 0.578 & \ul{0.720} & \textbf{0.760} \\
 &  & 40\% & 0.560 & 0.560 & 0.578 & \ul{0.880} & \textbf{0.940} \\
 &  & 60\% & 0.580 & 0.560 & 0.578 & \ul{0.920} & \textbf{0.960} \\
 &  & 80\% & 0.580 & 0.560 & 0.600 & \ul{0.980} & \textbf{1.000} \\
 &  & 100\% & 0.580 & 0.560 & 0.600 & \textbf{1.000} & \textbf{1.000} \\ \cline{2-8} 
 & \multirow{6}{*}{\begin{tabular}[c]{@{}c@{}}$\mathsf{ER}$\\ (10\%)\end{tabular}} & 0\% & 0.000 & 0.000 & 0.000 & 0.000 & 0.000 \\
 &  & 20\% & 0.000 & 0.000 & 0.000 & \ul{0.060} & \textbf{0.280} \\
 &  & 40\% & 0.000 & 0.000 & 0.000 & \ul{0.260} & \textbf{0.360} \\
 &  & 60\% & 0.000 & 0.000 & 0.000 & \textbf{0.460} & \ul{0.360} \\
 &  & 80\% & 0.000 & 0.000 & 0.000 & \textbf{0.660} & \ul{0.600} \\
 &  & 100\% & 0.000 & 0.000 & 0.000 & \textbf{0.820} & \ul{0.740} \\ \hline
\multirow{12}{*}{Magzine} & \multirow{6}{*}{\begin{tabular}[c]{@{}c@{}}$\mathsf{ER}$\\ (5\%)\end{tabular}} & 0\% & 0.740 & 0.740 & 0.740 & 0.740 & 0.740 \\
 &  & 20\% & 0.760 & 0.740 & 0.760 & \ul{0.760} & \textbf{0.780} \\
 &  & 40\% & 0.760 & 0.760 & 0.760 & \textbf{0.880} & \ul{0.860} \\
 &  & 60\% & 0.760 & 0.760 & 0.760 & \textbf{0.920} & \ul{0.880} \\
 &  & 80\% & 0.760 & 0.760 & 0.760 & \textbf{0.960} & \ul{0.880} \\
 &  & 100\% & 0.780 & 0.760 & 0.760 & \textbf{0.980} & \ul{0.880} \\ \cline{2-8} 
 & \multirow{6}{*}{\begin{tabular}[c]{@{}c@{}}$\mathsf{ER}$\\ (10\%)\end{tabular}} & 0\% & 0.380 & 0.380 & 0.380 & 0.380 & 0.380 \\
 &  & 20\% & 0.380 & 0.380 & 0.380 & \ul{0.500} & \textbf{0.600} \\
 &  & 40\% & 0.400 & 0.380 & 0.380 & \ul{0.560} & \textbf{0.740} \\
 &  & 60\% & 0.400 & 0.380 & 0.400 & \ul{0.620} & \textbf{0.760} \\
 &  & 80\% & 0.400 & 0.380 & 0.400 & \ul{0.760} & \textbf{0.820} \\
 &  & 100\% & 0.400 & 0.400 & 0.400 & \ul{0.840} & \textbf{0.860} \\ \hline
\end{tabular}
\end{table}

\begin{table}[htb]
\centering
\setlength{\tabcolsep}{2.5pt}
\caption{Graph attack results on \textit{ProxGraphRW} model.}
\label{tab:G-attack_pro}
\begin{tabular}{ccllllll}
\hline
\multicolumn{1}{l}{Dataset} & \multicolumn{1}{l}{Similarity} & budget & RndAdd & DegAdd & DeepWalk & alterI & cf \\ \hline
\multirow{14}{*}{\begin{tabular}[c]{@{}c@{}}KDD99\\ $\mathsf{ER}$ \\(5\%)\end{tabular}} & \multirow{7}{*}{cosine} & 0\% & 0.045 & 0.045 & 0.045 & 0.045 & 0.045 \\
 &  & 10\% & 0.045 & 0.045 & 0.045 & \ul{0.050} & \textbf{0.055} \\
 &  & 20\% & 0.045 & 0.045 & 0.045 & \ul{0.155} & \textbf{0.245} \\
 &  & 40\% & 0.045 & 0.045 & 0.045 & \ul{0.605} & \textbf{0.620} \\
 &  & 60\% & 0.045 & 0.045 & 0.045 & \ul{0.745} & \textbf{0.825} \\
 &  & 80\% & 0.055 & 0.045 & 0.050 & \ul{0.775} & \textbf{0.865} \\
 &  & 100\% & 0.085 & 0.045 & 0.060 & \ul{0.775} & \textbf{0.875} \\ \cline{2-8} 
 & \multirow{7}{*}{correlation} & 0\% & 0.045 & 0.045 & 0.045 & 0.045 & 0.045 \\
 &  & 10\% & 0.045 & 0.045 & 0.045 & \ul{0.055} & \textbf{0.060} \\
 &  & 20\% & 0.045 & 0.045 & 0.045 & \ul{0.110} & \textbf{0.150} \\
 &  & 40\% & 0.045 & 0.045 & 0.045 & \ul{0.315} & \textbf{0.405} \\
 &  & 60\% & 0.045 & 0.045 & 0.045 & \ul{0.575} & \textbf{0.690} \\
 &  & 80\% & 0.050 & 0.045 & 0.050 & \ul{0.670} & \textbf{0.735} \\
 &  & 100\% & 0.060 & 0.045 & 0.055 & \ul{0.695} & \textbf{0.845} \\ \hline
\multirow{14}{*}{\begin{tabular}[c]{@{}c@{}}MNIST\\ $\mathsf{ER}$\\ (10\%)\end{tabular}} & \multirow{7}{*}{cosine} & 0\% & 0.000 & 0.000 & 0.000 & 0.000 & 0.000 \\
 &  & 10\% & 0.000 & 0.000 & 0.000 & \textbf{0.060} & \ul{0.045} \\
 &  & 20\% & 0.000 & 0.000 & 0.000 & \textbf{0.210} & \ul{0.135} \\
 &  & 40\% & 0.000 & 0.000 & 0.000 & \textbf{0.585} & \ul{0.515} \\
 &  & 60\% & 0.000 & 0.000 & 0.020 & \ul{0.800} & \textbf{0.860} \\
 &  & 80\% & 0.005 & 0.000 & 0.030 & \ul{0.940} & \textbf{0.975} \\
 &  & 100\% & 0.050 & 0.000 & 0.070 & \ul{0.985} & \textbf{0.995} \\ \cline{2-8} 
 & \multirow{7}{*}{correlation} & 0\% & 0.000 & 0.000 & 0.000 & 0.000 & 0.000 \\
 &  & 10\% & 0.000 & 0.000 & 0.000 & \ul{0.045} & \textbf{0.060} \\
 &  & 20\% & 0.000 & 0.000 & 0.000 & \textbf{0.205} & \ul{0.185} \\
 &  & 40\% & 0.000 & 0.000 & 0.000 & \textbf{0.555} & \ul{0.550} \\
 &  & 60\% & 0.010 & 0.000 & 0.010 & \ul{0.770} & \textbf{0.825} \\
 &  & 80\% & 0.040 & 0.000 & 0.045 & \textbf{0.940} & \textbf{0.940} \\
 &  & 100\% & 0.095 & 0.005 & 0.080 & \textbf{0.995} & \textbf{0.995} \\ \hline
\end{tabular}
\vspace{-5pt}
\end{table}

Comparing \textbf{alterI} and \textbf{cf} attack, it was observed that \textbf{cf} attack slightly outperforms \textbf{alterI} in most cases. In our experiments, we observe that \textbf{cf} can achieve significantly lower attack loss in the continuous domain (i.e., $\tilde{B}$). However, when discretizing the optimization results, the attack performance is not guaranteed to be preserved. While \textbf{cf} is generally more effective (also observed for feature-space attacks in Section~\ref{featureAttackResults}), \textbf{alterI} is more efficient on larger graphs such as KDD-99 and MNIST (see Tab.~\ref{tab:runtime}).

\subsubsection{Preferences of graph attack}
We further present a more detailed analysis of the graph attack results in Fig.~\ref{fig:graph_attack_analysis}, in which Fig.~\ref{fig:analysis_a} and \ref{fig:analysis_b} show the proportion of the attacked nodes (to the total number of nodes) corresponding to different budgets. On average, only about $1\%-6\%$ (KDD-99) and $0.3\%-2.7\%$ (MNIST) of nodes are involved in the edge modification under various budgets (Fig.~\ref{fig:analysis_b}).
In Fig.~\ref{fig:analysis_c}, we present the node degrees of attack nodes and others, and we observe that the attacker prefers nodes with lower degrees as attack nodes.
Fig.~\ref{fig:analysis_d} presents the weights changed in the attack. We observe that when the budget is limited, the attacker tends to delete the edges with the original weights
close to $1.0$ or add edges with original weights close to $0$. The attacker mainly adds/deletes edges between target nodes and other nodes (Fig.~\ref{fig:analysis_e}), and the target-other edge modification tends to increase the degree of target nodes (Fig.~\ref{fig:analysis_f}). These actually provide convenience for our graph-guided feature attack with attack loss $\mathcal{L}_g(\tilde{\mathbf{X}})$, where the target node features are fixed (the edges between target-target are fixed) and the attack nodes can be optimized to be close to the desired edge weights (the edges between target nodes and control nodes). We observe similar phenomena in the MNIST dataset. 

\begin{figure}[htb]
\centering
    \subfigure[Attack nodes (KDD-99).]{
    \includegraphics[width=0.22\textwidth,height=2.5cm]{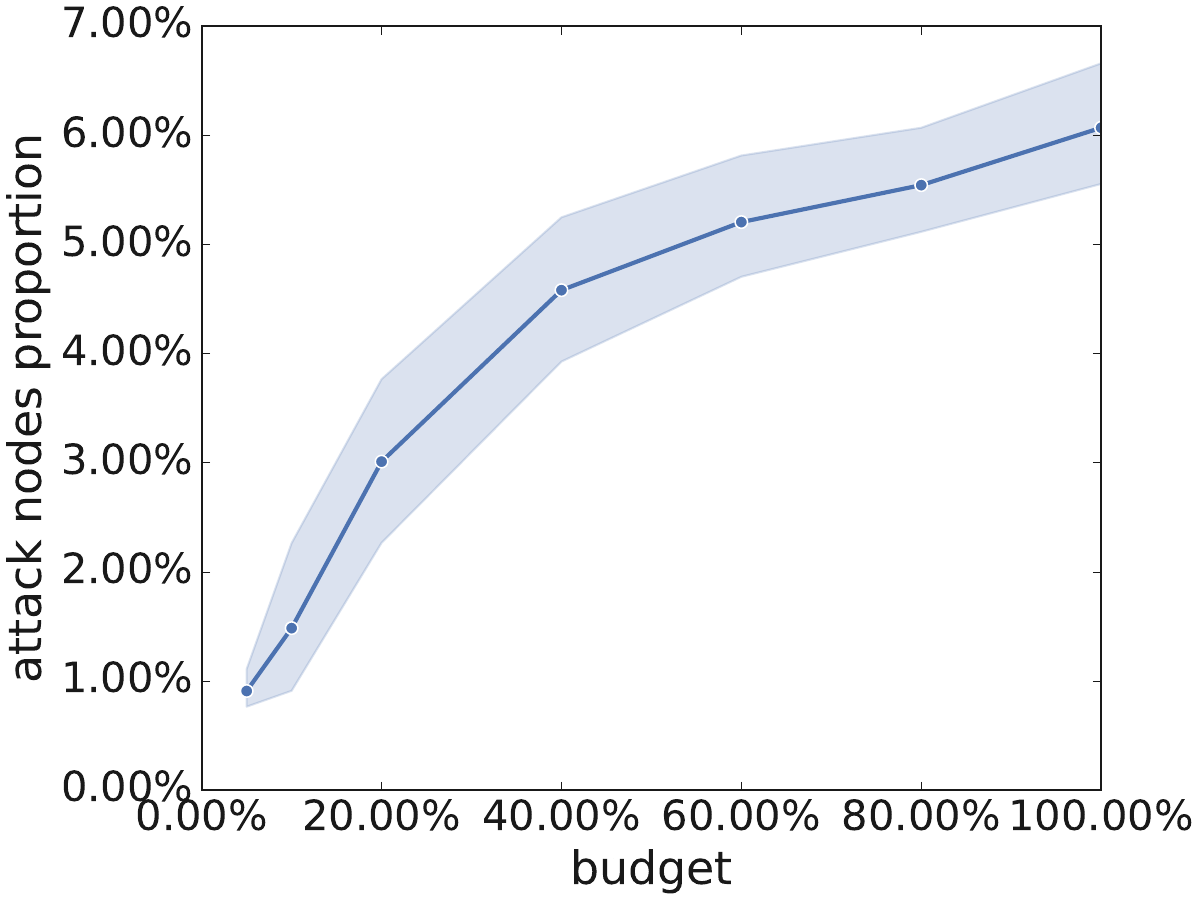}\label{fig:analysis_a}
    }
    \subfigure[Attack nodes (MNIST).]{
    \includegraphics[width=0.22\textwidth,height=2.5cm]{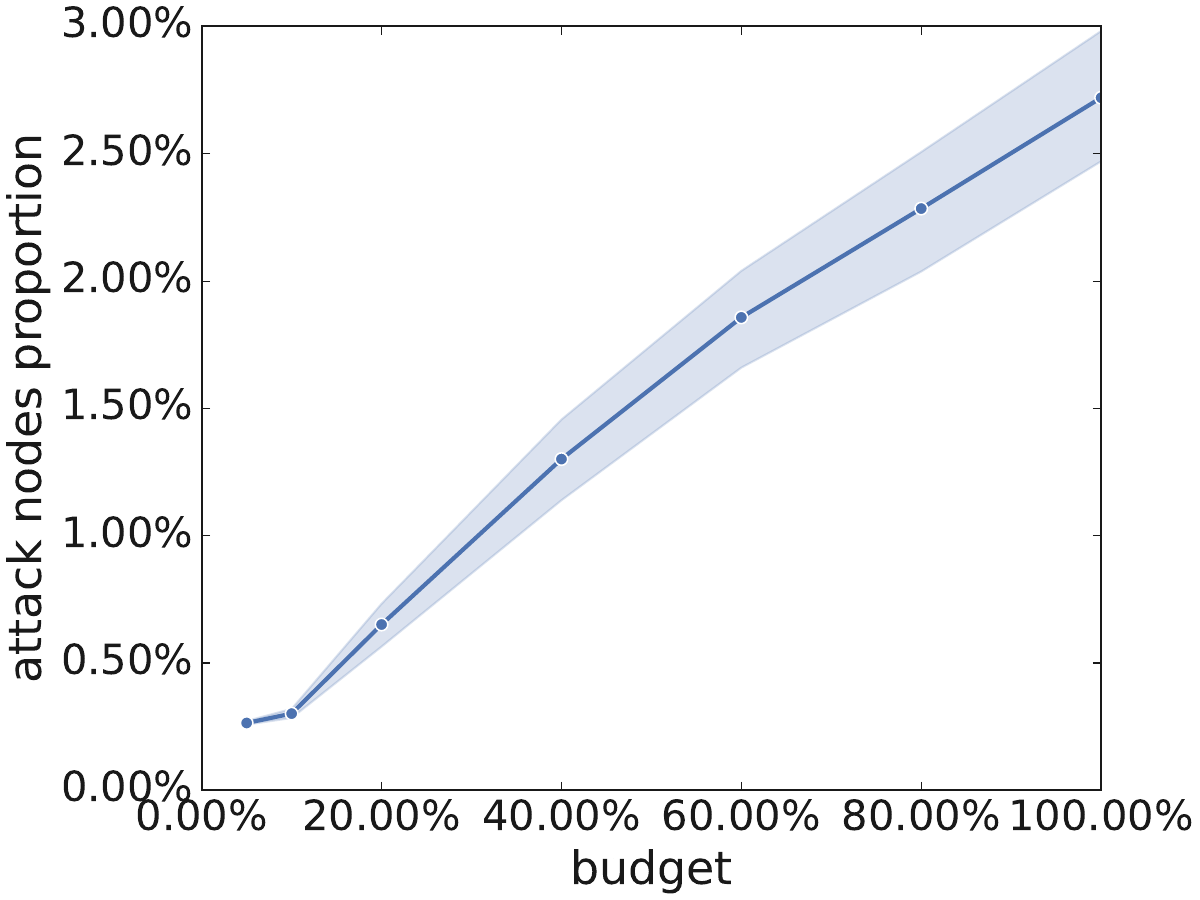}\label{fig:analysis_b}
    }
    \subfigure[The attacker prefers nodes with lower degrees.]{
    \includegraphics[width=0.22\textwidth,height=2.5cm]{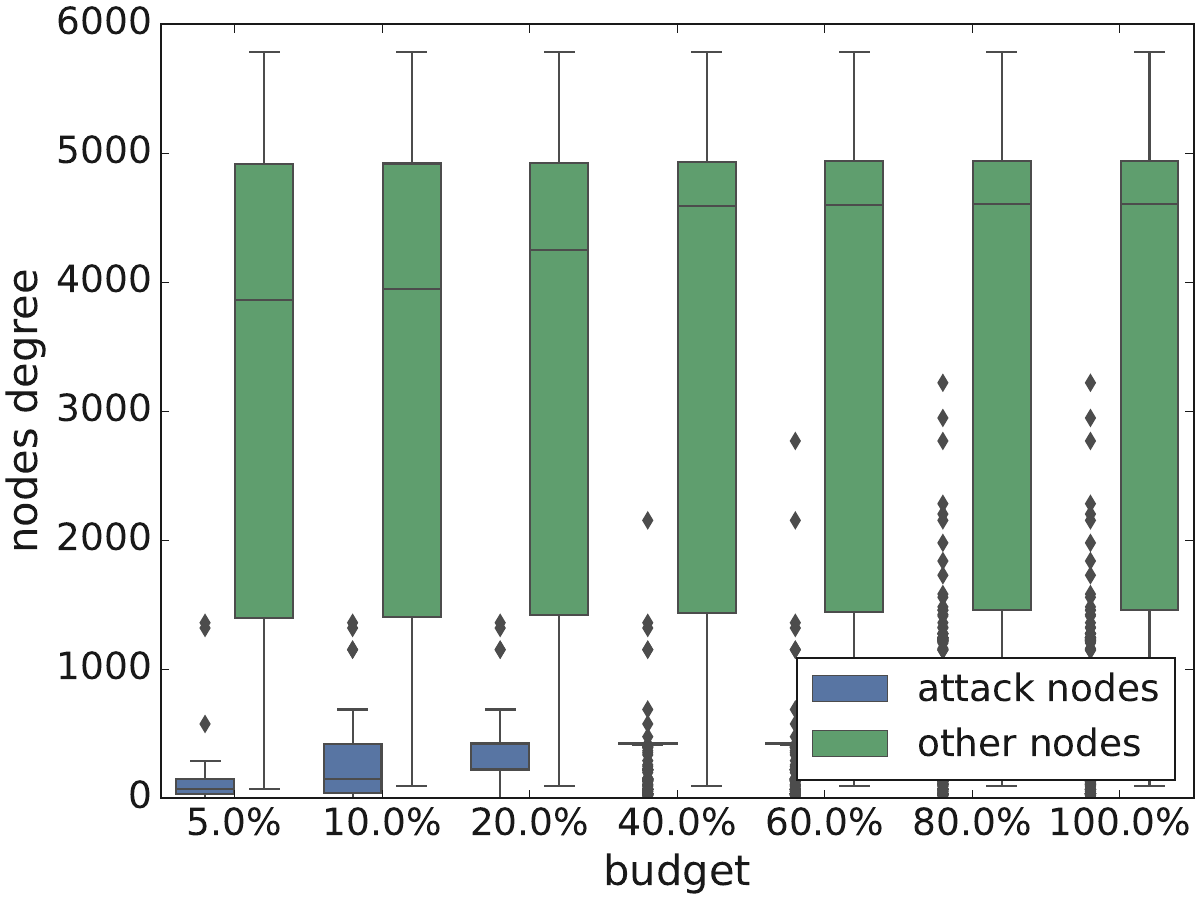}
    \label{fig:analysis_c}}
    \subfigure[The weights changed by attacker.]{
    \includegraphics[width=0.22\textwidth,height=2.5cm]{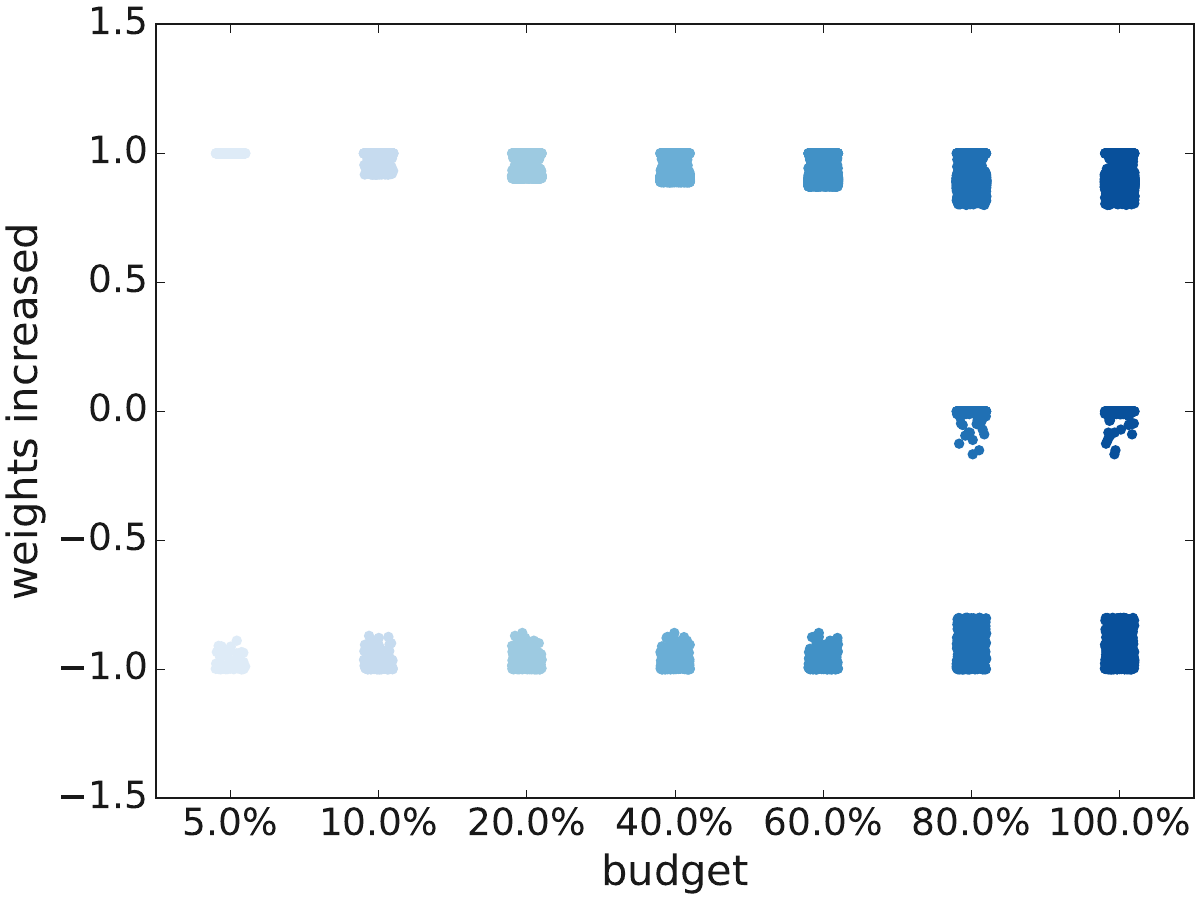}
    \label{fig:analysis_d}}
    \subfigure[The attacker prefers adding edges between target nodes and the others.]{
    \includegraphics[width=0.22\textwidth,height=2.5cm]{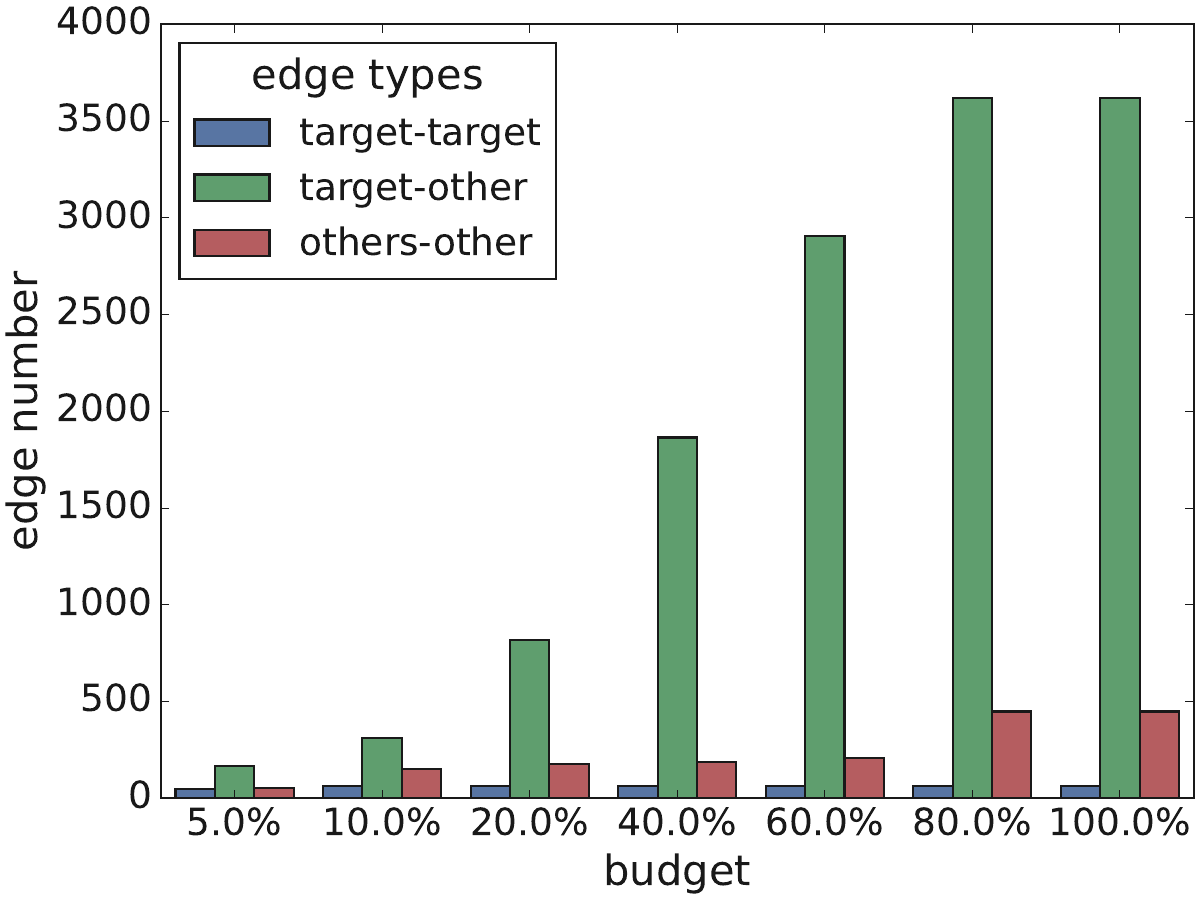}
    \label{fig:analysis_e}}
    \subfigure[The attacker increases the degree of target nodes.]{
    \includegraphics[width=0.22\textwidth,height=2.5cm]{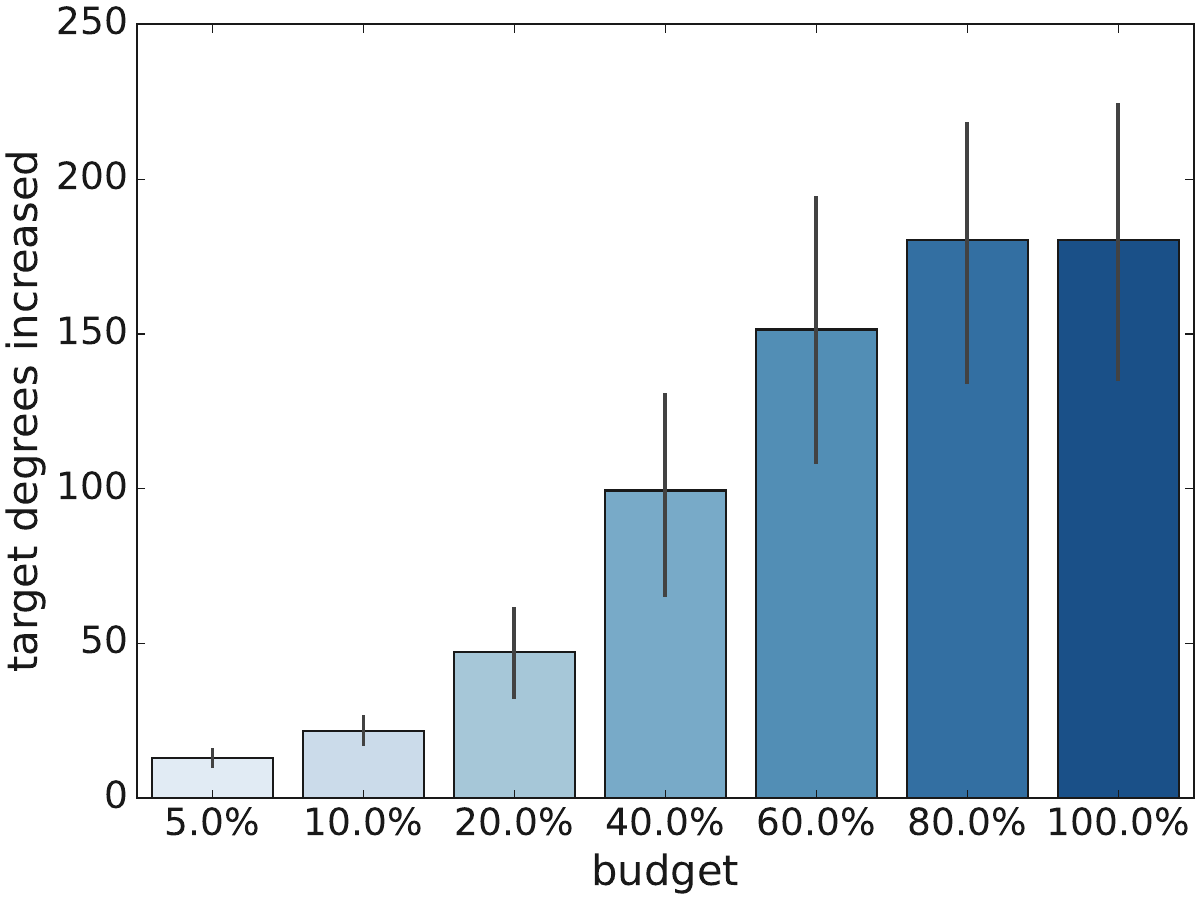}
    \label{fig:analysis_f}}
 \caption{Graph-space attack (\textbf{alterI}) result analysis on KDD-99 dataset.}
\label{fig:graph_attack_analysis}
\end{figure}

\begin{figure}[htb]
\centering
  \subfigure[KDD-99]{
    \includegraphics[width=0.22\textwidth,height=2.5cm]{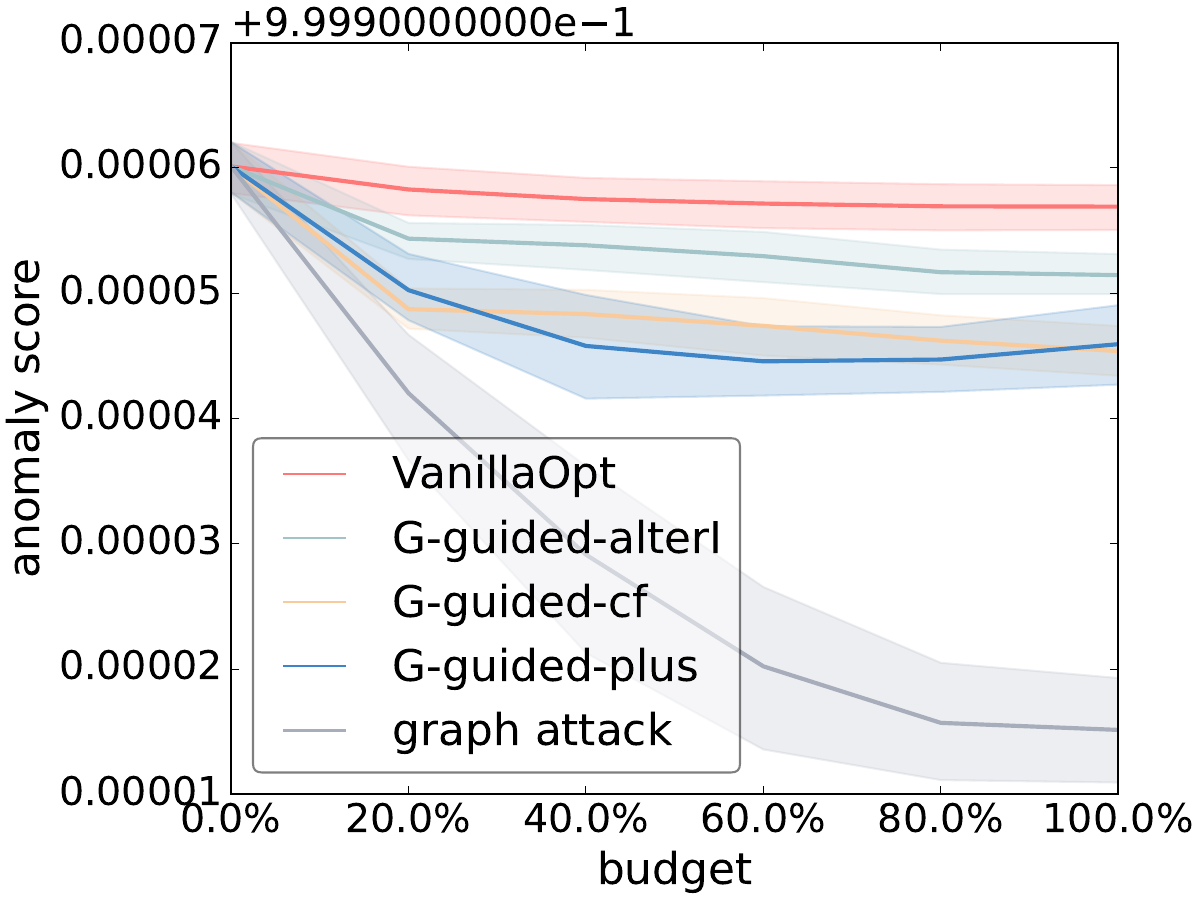}\label{fig:feature_attack_a}
    }
    \subfigure[KDD-99]{
    \includegraphics[width=0.22\textwidth,height=2.5cm]{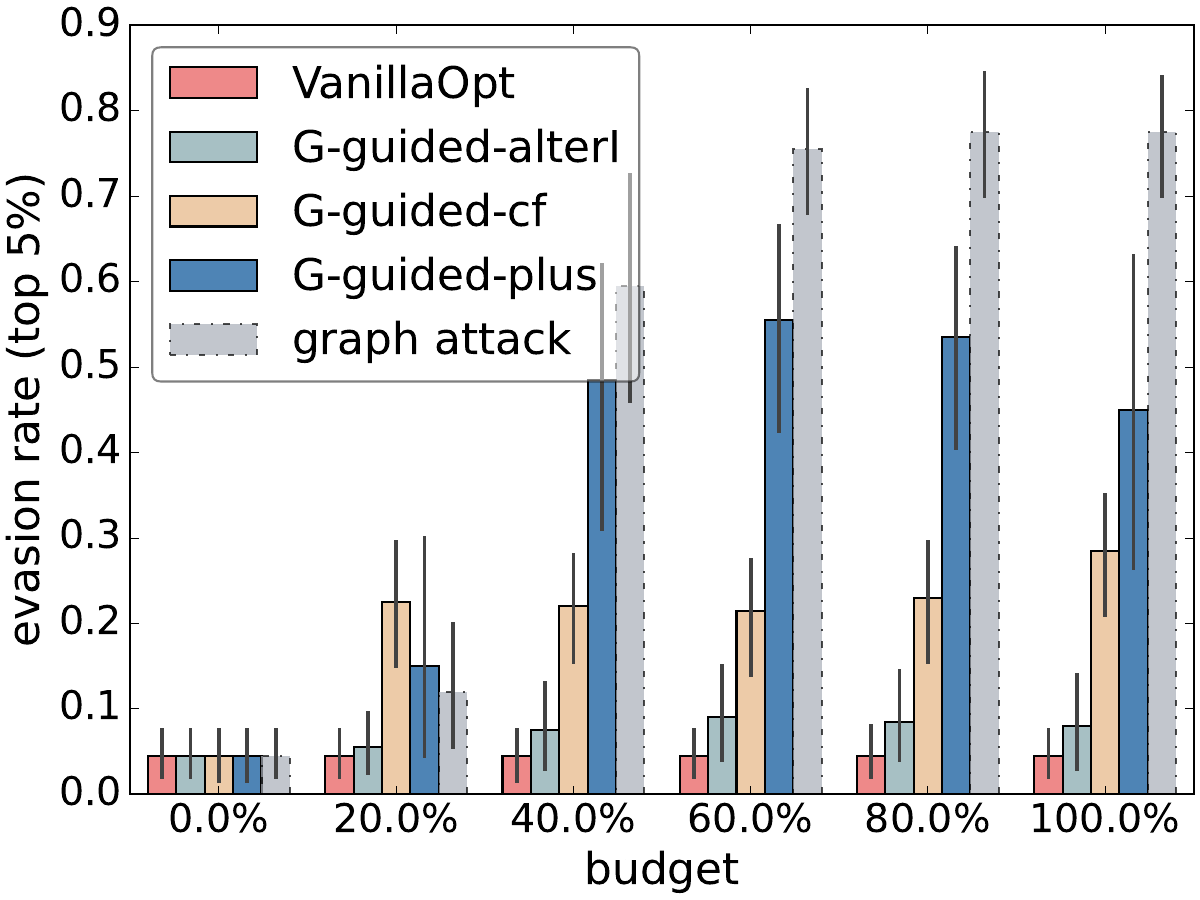}\label{fig:feature_attack_b}
    }
    \subfigure[MNIST]{
    \includegraphics[width=0.22\textwidth,height=2.5cm]{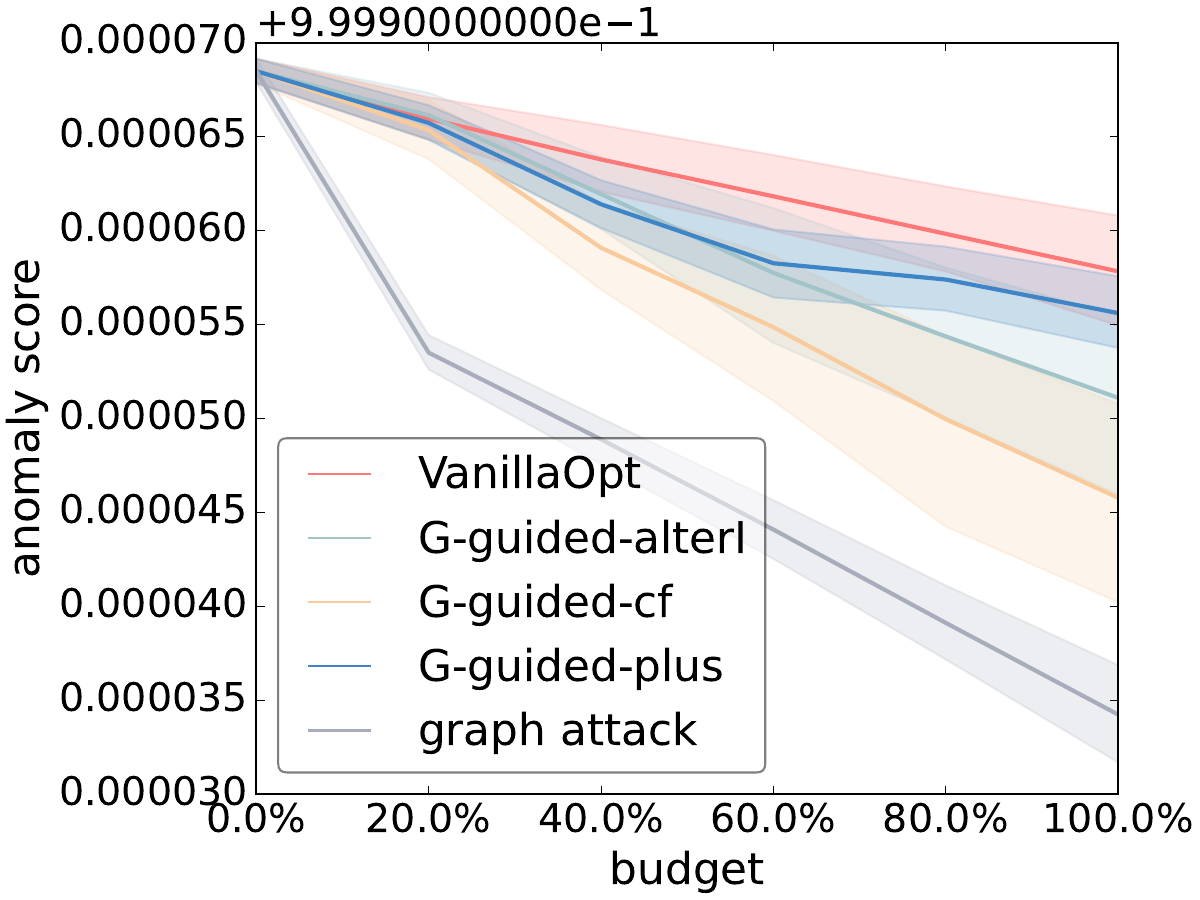}\label{fig:feature_attack_c}
    }
    \subfigure[MNIST]{
    \includegraphics[width=0.22\textwidth,height=2.5cm]{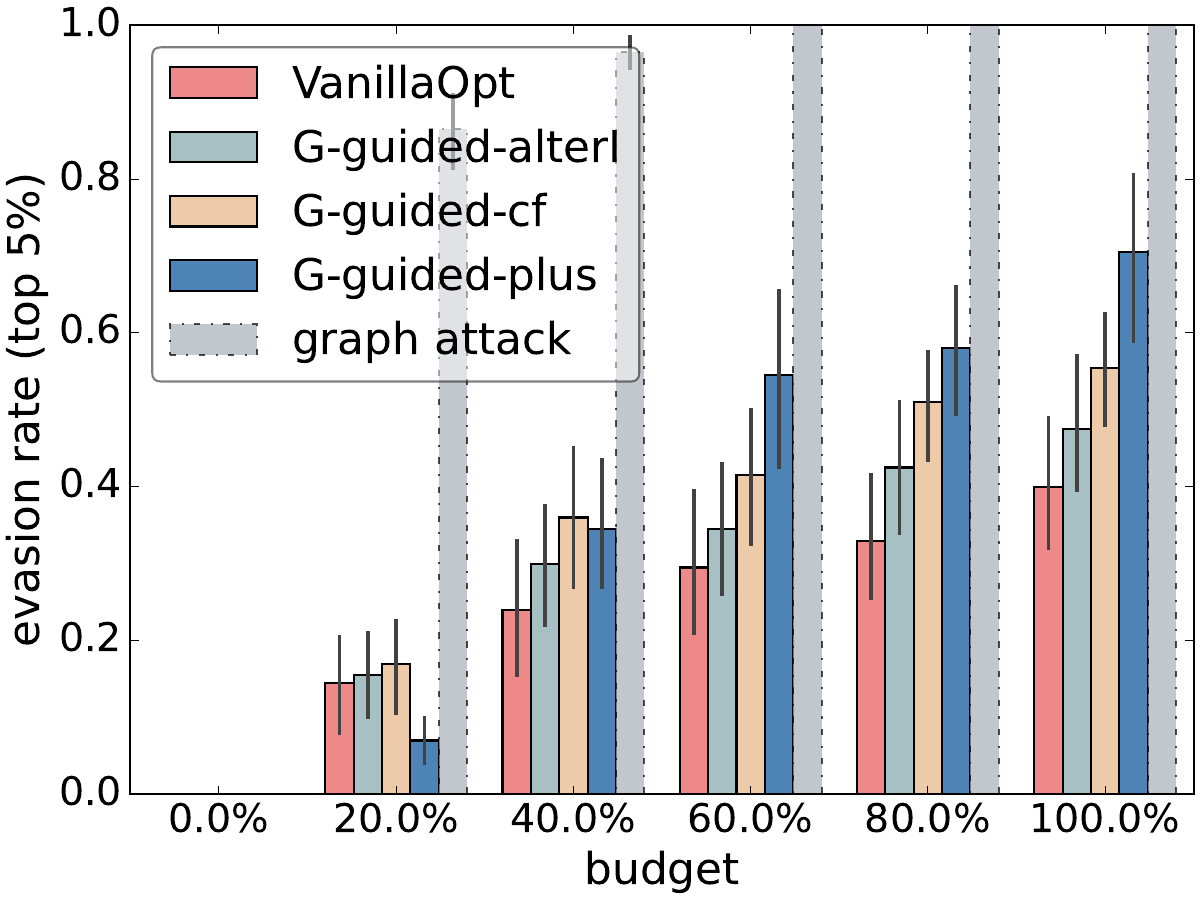}\label{fig:feature_attack_d}
    }
 \caption{Feature-space attack results.}
\label{fig:feature_attack}
\end{figure}

\begin{figure}[htb]
\centering
  \subfigure[Control node degrees (KDD-99)]{
  \includegraphics[width=0.22\textwidth,height=2.5cm]{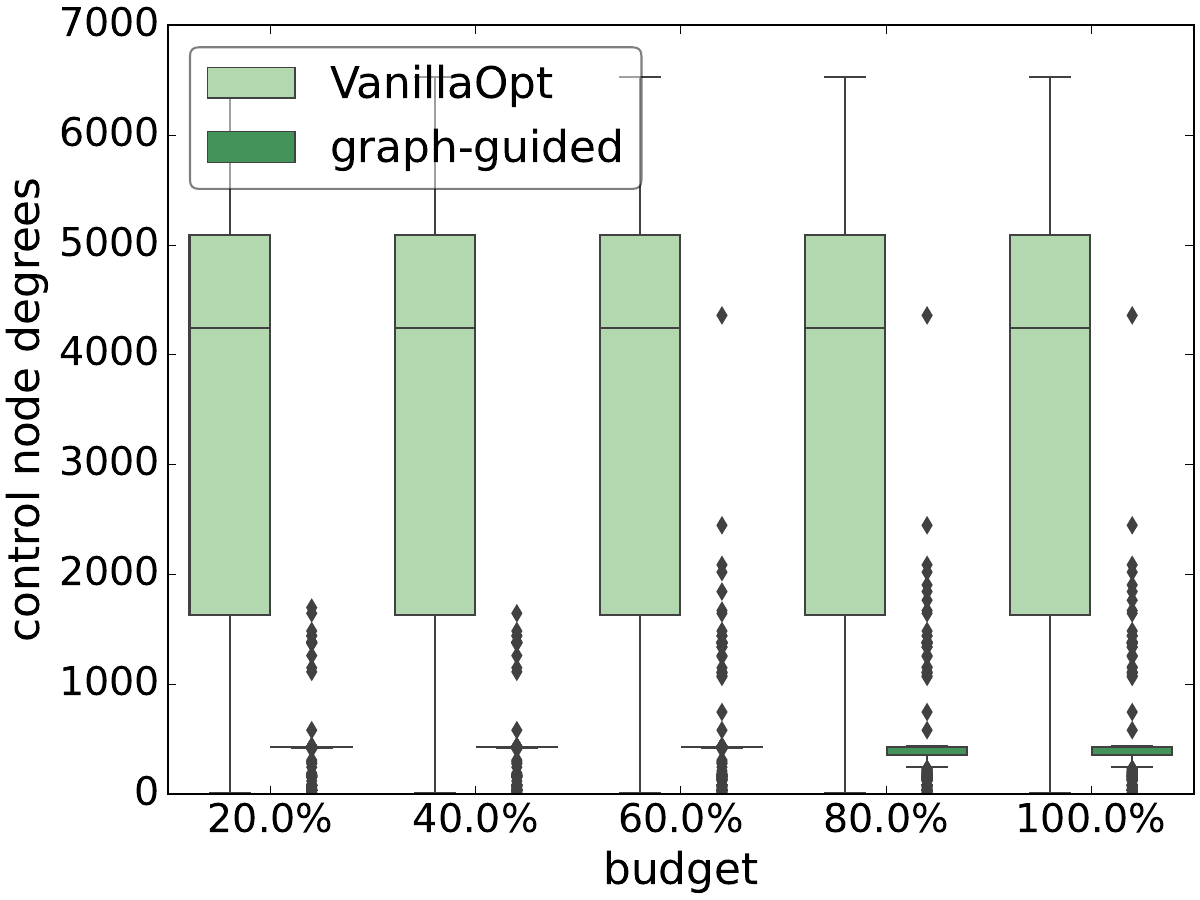}\label{fig:feature_analysis_a}}
    \subfigure[Edge modified (KDD-99)]{
    \includegraphics[width=0.22\textwidth,height=2.5cm]{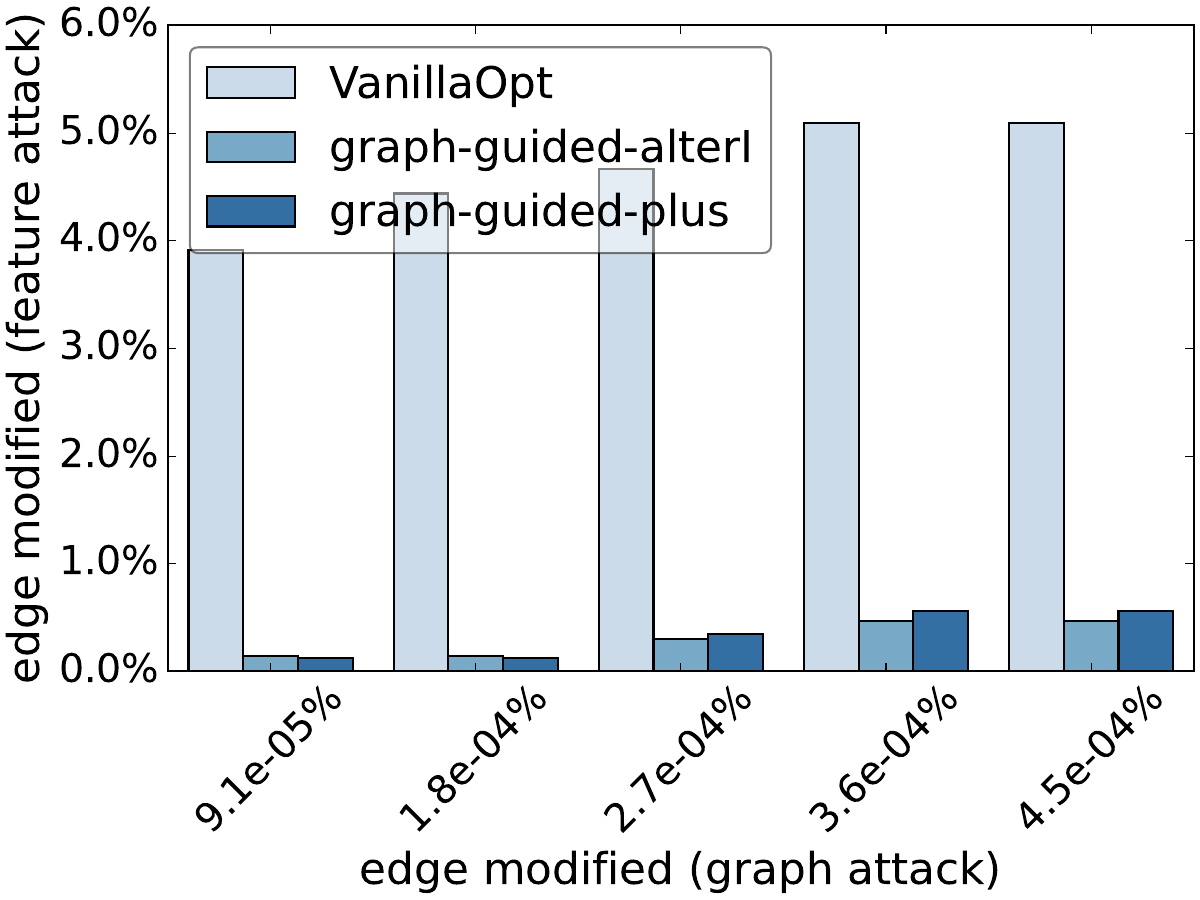}\label{fig:feature_analysis_b}}
 \caption{Result analysis of feature-space attacks.} 
\label{fig:feature_analysis}
\end{figure}

\subsection{Performances of Graph-guided Feature-space Attacks}
\label{featureAttackResults}
\subsubsection{Effectiveness of attacks}
We compare the performances of the feature-space attacks in Fig.~\ref{fig:feature_attack}. 
Our analysis shows that \textbf{G-guided-alterI} outperforms the \textbf{VanillaOpt} method, achieving much lower anomaly scores and higher evasion rates. These two models are only different in the selection of attack nodes, which indicates the effectiveness of using guidance from graph-space attacks in node selection. 
Comparing the performance of the \textbf{alterI} and \textbf{cf} attack strategies under $\mathcal{L}_a$, we observe that \textbf{cf}-attack also improves the performance, although the side effect is that \textbf{cf}-attack takes about 7 times longer than \textbf{alterI} in our experiments (Tab.~\ref{tab:runtime}). 
Additionally, \textbf{G-guided-plus} has a higher evasion rate than \textbf{G-guided-alterI} and \textbf{G-guided-cf} in most cases, indicating the advantage of using the attack loss $\mathcal{L}_g$ as further guidance for feature attack.

\subsubsection{Unnoticeability of attack}
In Fig.~\ref{fig:feature_analysis}, we provide an analysis of the feature attack highlighting its advantage of unnoticeability. As mentioned earlier, the graph attack prefers the attack nodes with lower degrees. As a result, our graph-guided attack nodes have lower node degrees compared to \textbf{VanillaOpt} (Fig.~\ref{fig:feature_analysis_a}). This leads to significantly fewer edge modifications in graph-guided attacks compared to \textbf{VanillaOpt} (Fig.~\ref{fig:feature_analysis_b}), which enhances the unnoticeability of the attack.


\begin{table}[htb]
\caption{Runtime comparison of \textbf{alterI} and \textbf{cf}-attack.}
\label{tab:runtime}
\centering
\setlength{\tabcolsep}{4.0pt}
\begin{tabular}{llcccc}
\hline
 & Attacks & Author-Paper & Magazine & KDD-99 & MNIST \\ \hline
\multirow{2}{*}{\begin{tabular}[c]{@{}l@{}}Graph\\  attack\end{tabular}} & alterI & 00:00:07 & 00:00:04 & 00:00:10 & 00:00:16 \\
 & cf & 00:00:02 & 00:00:02 & 00:00:23 & 00:00:35 \\ \hline
\multicolumn{1}{c}{\multirow{2}{*}{\begin{tabular}[c]{@{}c@{}}Feature\\ attack\end{tabular}}} & alterI & - & - & 00:00:27 & 00:00:18 \\
\multicolumn{1}{c}{} & cf & - & - & 00:03:26 & 00:01:49 \\ \hline
\end{tabular}
\vspace{-5pt}
\end{table}

\subsection{Transferability of graph-guided attack}
We transfer our feature-space attacks to several unsupervised anomaly detection models, including Beta-VAE~\cite{burgess2018understanding}, IForest~\cite{liu2008isolation}, and ECOD~\cite{li2022ecod}. 
Tab.~\ref{tab:transferAttack} shows the anomaly scores of target nodes before and after the transfer attack based on our \textbf{G-guided-alterI} and \textbf{G-guided-plus} feature attack on the KDD-99 dataset. The results indicate that the graph-guided attack with graph attack loss significantly decreases the anomaly scores of the target nodes across different models. This suggests that the graph-guided attack on RWAD has the potential to be used as a surrogate model for black-box attacks. The graph-guided attack could be a useful tool for attackers to evade detection and deceive anomaly detection systems in real-world scenarios.

\begin{table}[htb]
\centering
\caption{Transferability: The change in anomaly score (\%) compared to the clean data. Lower is better. }
\label{tab:transferAttack}
\setlength{\tabcolsep}{3.0pt}
\scalebox{0.98}{
    \begin{tabular}{ccccccc}
    \hline
        Detect Methods             &Attack Methods    & 20\%  & 40\%  & 60\%  & 80\%  & 100\% \\ \hline
        \multirow{3}{*}{Beta-VAE} & VanillaOpt           & -11.56 & -13.97 & -14.70 & -15.18 & -15.68 \\
        \multirow{3}{*}{}         & G-guided-alterI & -4.15 &  -5.65 &  -6.13 &  -7.14 & -9.711  \\
        \multirow{3}{*}{}         & G-guided-plus & \textbf{-25.26} & \textbf{-31.79} & \textbf{-33.25} & \textbf{-33.94} & \textbf{-33.99} \\ \hline
        
        \multirow{3}{*}{IForest}  & VanillaOpt                           & -10.63 & -0.06 & -8.41 & -0.82 & -11.93 \\
        \multirow{3}{*}{}         & G-guided-alterI      & 9.94 & 10.44 & -0.18 &  0.39 & -2.31  \\
        \multirow{3}{*}{}         & G-guided-plus & \textbf{-25.03} & \textbf{-44.21} & \textbf{-40.27} & \textbf{-47.58} & \textbf{-47.26} \\ \hline
        
        \multirow{3}{*}{ECOD}     & VanillaOpt                           & -2.20 & -2.72 & -2.90 & -2.988 & -3.099 \\
        \multirow{3}{*}{}         & G-guided-alterI      & -0.29 & -0.64 & -0.66 & -0.90 & -1.28 \\
        \multirow{3}{*}{}         & G-guided-plus & \textbf{-3.70} & \textbf{-5.32} & \textbf{-5.81} & \textbf{-6.01} & \textbf{-6.00} \\ \hline
    \end{tabular}
}
\end{table}

\section{Conclusion}
\label{conclusion}
In conclusion, this paper has shed light on the vulnerabilities of Random-Walk-based Anomaly Detection (RWAD), a classical and important anomaly detection tool. Specifically, we introduce a novel study of adversarial poisoning attacks on RWAD, where the graph is constructed on top of the feature space. We provide a theoretical understanding of these attacks, including proof of NP-hardness. Our approach involves proposing graph-space attacks and using the graph attack to guide the feature-space attack, which bridges the gap between these two attacks. Our experiments on four datasets, encompassing both directly and indirectly accessible graphs, demonstrate the effectiveness of our proposed graph-space attack and its ability to guide the selection of attack nodes and optimization of the attack loss for feature-space attacks. By taking RWAD as an example, our study provides valuable insights into the effectiveness of graph-space attacks and feature-space attacks. Future research can extend this work to apply RWAD for black-box attacks on other deep learning-based anomaly detection systems, without relying on labeled data or inner models. 

\balance
\bibliographystyle{IEEEtran} 	
\bibliography{IEEEabrv,references}
\end{document}